\documentclass[sigconf, 9pt]{acmart}

\setcopyright{rightsretained}

\usepackage[subtle, mathspacing=normal]{savetrees}

\usepackage{xspace}
\usepackage{amsmath}
\usepackage{amsthm}
\usepackage{amstext}

\usepackage[normalem]{ulem} 

\usepackage{tikz}
\usepackage{cleveref}
\usepackage{subcaption}
\usepackage{pgfplots}
\pgfplotsset{compat=1.3}
\usepackage{algorithm}
\usepackage{algpseudocode}
\usepackage{wrapfig}
\usepackage{url}
\usepackage{enumerate}
\usepackage{enumitem}
\usepackage[normalem]{ulem}
\usepackage{relate}
\usepackage{marginnote}
\usepackage[textwidth=0.75in,textsize=tiny]{todonotes} 
\usepackage{siunitx}

\input{defines}

\title{\sysname: High Performance PMEM Hash Tables\\ Through Stability and Low Associativity}
\author{Prashant Pandey}
\email{pandey@cs.utah.edu}
\affiliation{University of Utah}
\author{Michael A.~Bender}
\email{bender@cs.stonybrook.edu}
\affiliation{Stony Brook University}
\author{Alex Conway}
\email{aconway@vmware.com}
\affiliation{VMware Research}
\author{Mart\'{\i}n Farach-Colton}
\email{farach@rutgers.edu}
\affiliation{Rutgers University}
\author{William Kuszmaul}
\email{kuszmaul@mit.edu}
\affiliation{MIT}
\author{Guido Tagliavini}
\email{guido.tag@rutgers.edu}
\affiliation{Rutgers University}
\author{Rob Johnson}
\email{robj@vmware.com}
\affiliation{VMware Research}

\begin{document}
\fancyhead{}


\begin{abstract}
  Modern hash table designs strive to minimize space while maximizing
  speed. The most important factor in speed is the number of cache
  lines accessed during updates and queries.  This is especially
  important on PMEM, which is slower than DRAM and in which writes are
  more expensive than reads.

  This paper proposes two stronger design objectives: stability and
  low-associativity.  A stable hash table doesn't move items around,
  and a hash table has low associativity if there are only a few
  locations where an item can be stored.  Low associativity ensures
  that queries need to examine only a few memory locations, and
  stability ensures that insertions write to very few cache lines.
  Stability also simplifies scaling and crash safety.

  We present \sysname, a fast, crash-safe, concurrent, and
  space-efficient hash table for PMEM based on the design principles
  of stability and low associativity.  \Sysname \rev{R4D9}{combines in-memory metadata} with a new
  hashing technique, iceberg hashing, that is (1) space efficient, (2) stable,
  and (3) supports low associativity. \rev{R4D21}
  In contrast, existing hash-tables either modify numerous
  cache lines during insertions (e.g. cuckoo hashing), access numerous
  cache lines during queries (e.g. linear probing), or waste space
  (e.g. chaining).  Moreover, the combination of (1)-(3) yields
  several emergent benefits: \sysname scales better than other hash
  tables, supports crash-safety, and has excellent performance on PMEM
  (where writes are particularly expensive).

  In our benchmarks, \sysname inserts are 50\% to $3\times$ faster
  than state-of-the-art PMEM hash tables Dash and CLHT and queries are
  20\% to $2\times$ faster.  \sysname space overhead is 17\%, whereas
  Dash and CLHT have space overheads of $2\times$ and $3\times$,
  respectively.  \Sysname also exhibits linear scaling and is crash
  safe.  In DRAM, \sysname outperforms state-of-the-art hash tables
  libcuckoo and CLHT by almost $2\times$ on insertions while offering
  good query throughput and much better space efficiency.

\end{abstract}

\maketitle


\section{Introduction}


\newcommand{\figwidth}{0.49\textwidth}
\pgfplotsset{
  GenericLinePlot/.style={
  },
}

\pgfplotsset{
  ResizeRSSPlot/.style={
    small,
    width = \linewidth + 15pt,
    height=0.75\linewidth,
    xlabel={Number of insertions (M)},
    ylabel={Max RSS (GB)},
    xlabel near ticks,
    x label style={at={(0.5,-0.12)},font=\small},
    scaled ticks=false,
    scaled x ticks = false,
    ylabel near ticks,
    y label style={at={(-0.1,0.5)},font=\small},
    ytick style={draw=none},
    scaled y ticks = base 10:-6,
    ytick scale label code/.code = {},
    yticklabel style={ /pgf/number format/fixed },
    ymajorgrids,
    yminorgrids,
    ymin=0,
    minor tick num=1,
    minor grid style={draw=gray!25},
  },
}

\pgfplotsset{
  RSSPlot/.style={
    small,
    width = \linewidth + 15pt,
    height=0.75\linewidth,
    xlabel={Load factor},
    ylabel={Max RSS (GB)},
    xlabel near ticks,
    x label style={at={(0.5,-0.12)},font=\small},
    scaled ticks=false,
    ylabel near ticks,
    y label style={at={(-0.1,0.5)},font=\small},
    ytick style={draw=none},
    scaled y ticks = base 10:-6,
    ytick scale label code/.code = {},
    yticklabel style={ /pgf/number format/fixed },
    ymajorgrids,
    yminorgrids,
    ymin=0,
    minor tick num=1,
    minor grid style={draw=gray!25},
  },
}

\pgfplotsset{
  InstanSePlot/.style={
    small,
    width = \linewidth,
    height=0.6\linewidth,
    xlabel={Space efficiency},
    ylabel={Throughput (M/s)},
    xlabel near ticks,
    x label style={at={(0.5,-0.12)},font=\small},
    scaled ticks=false,
    ylabel near ticks,
    y label style={at={(-0.1,0.5)},font=\small},
    ytick style={draw=none},
    scaled y ticks = base 10:-6,
    ytick scale label code/.code = {},
    yticklabel style={ /pgf/number format/fixed },
    ymajorgrids,
    yminorgrids,
    ymax=20000000,
    ymin=0,
    minor tick num=1,
    minor grid style={draw=gray!25},
  },
}

\pgfplotsset{
  InstanPlot/.style={
    small,
    width = \linewidth + 15pt,
    height=0.75\linewidth,
    xlabel={Load factor},
    ylabel={Throughput (M/s)},
    xlabel near ticks,
    x label style={at={(0.5,-0.12)},font=\small},
    scaled ticks=false,
    ylabel near ticks,
    y label style={at={(-0.1,0.5)},font=\small},
    ytick style={draw=none},
    scaled y ticks = base 10:-6,
    ytick scale label code/.code = {},
    yticklabel style={ /pgf/number format/fixed },
    ymajorgrids,
    yminorgrids,
    ymin=0,
    minor tick num=1,
    minor grid style={draw=gray!25},
  },
}

\pgfplotsset{
  HistoPlot/.style={
    footnotesize,
    ybar = 1pt,
    bar width = 3pt,
    width = \linewidth + 30pt,
    height = 1.8in,
    ymin = 0,
    xmin = 250,
    xmax = 20000,
    xtick=data,
    xticklabels={250\,\si{\nano\second}, 500\,\si{\nano\second}, 1\,\si{\micro\second}, 2.5\,\si{\micro\second}, 5\,\si{\micro\second}, 10\,\si{\micro\second}, >10\,\si{\micro\second}},
    x tick label style = {rotate = 45, anchor = 15, xshift = 2pt},
    ytick scale label code/.code = {},
    yticklabel style={
        /pgf/number format/fixed,
    },
    scaled y ticks = base 10:-6,
    enlarge x limits = 0.08,
    fill,
    no marks,
    nodes near coords,
    every node near coord/.append style={rotate = 90, anchor = west, font=\tiny,/pgf/number format/fixed relative,/pgf/number format/precision=2},
    nodes near coords={
       \pgfkeys{/pgf/fpu}%
       \pgfmathparse{\pgfplotspointmeta / 640000}%
       \pgfmathprintnumber{\pgfmathresult}\%
    },
  },
}

\pgfplotsset{
  IntroBarPlotWrite/.style={
    footnotesize,
    ybar = 2pt,
    bar width = 5pt,
    width = \linewidth + 30pt,
    height = 1.5in,
    no marks,
    ylabel ={Throughput(M/s)},
    ymin = 0,
    symbolic x coords = {ins, del},
    xticklabels = {{}, {}, Insertion, {}, {}, {}, {}, Deletion},
    x tick label style = {font = \small, align = center},
    major x tick style = transparent,
    enlarge x limits = 0.3,
    ymajorgrids,
    yminorgrids,
    scaled ticks=false,
    scaled y ticks = base 10:-6,
    ytick scale label code/.code = {},
    minor y tick num = 1,
    fill,
  },  
}

\pgfplotsset{
  IntroBarPlotRead/.style={
    footnotesize,
    ybar = 2pt,
    bar width = 5pt,
    width = \columnwidth + 30pt,
    height = 1.5in,
    no marks,
    ymin = 0,
    symbolic x coords = {pos, neg},
    xticklabels = {{}, {}, Pos Query, {}, {}, {}, {}, Neg Query},
    x tick label style = {font = \small, align = center},
    major x tick style = transparent,
    enlarge x limits = 0.3,
    ymajorgrids,
    yminorgrids,
    scaled ticks=false,
    scaled y ticks = base 10:-6,
    ytick scale label code/.code = {},
    minor y tick num = 1,
    fill,
  },  
}

\pgfplotsset{
  YcsbBarPlotLa/.style={
    footnotesize,
    ybar = 2pt,
    bar width = 5pt,
    width = \linewidth + 30pt,
    height = 1.5in,
    no marks,
    ymin = 0,
    symbolic x coords = {load, runa},
    xticklabels = {{}, {}, Load, {}, {}, {}, {}, Run A},
    x tick label style = {font = \small, align = center},
    major x tick style = transparent,
    enlarge x limits = 0.3,
    ymajorgrids,
    yminorgrids,
    scaled ticks=false,
    scaled y ticks = base 10:-6,
    ytick scale label code/.code = {},
    minor y tick num = 1,
    fill,
  },  
}

\pgfplotsset{
  YcsbBarPlotBc/.style={
    footnotesize,
    ybar = 2pt,
    bar width = 5pt,
    width = \columnwidth + 30pt,
    height = 1.5in,
    no marks,
    ymin = 0,
    symbolic x coords = {runb, runc},
    xticklabels = {{}, {}, Run B, {}, {}, {}, {}, Run C},
    x tick label style = {font = \small, align = center},
    major x tick style = transparent,
    enlarge x limits = 0.3,
    ymajorgrids,
    yminorgrids,
    scaled ticks=false,
    scaled y ticks = base 10:-6,
    ytick scale label code/.code = {},
    minor y tick num = 1,
    fill,
  },  
}

\pgfplotsset{
  LatencyPlot/.style={
    small,
    width = \linewidth + 15pt,
    height=0.75\linewidth,
    xlabel={Latency (ns)},
    ylabel={CDF},
    xlabel near ticks,
    x label style={at={(0.5,-0.12)},font=\small},
    scaled ticks=false,
    ylabel near ticks,
    y label style={at={(-0.1,0.5)},font=\small},
    ytick style={draw=none},
    ytick scale label code/.code = {},
    yticklabel style={ /pgf/number format/fixed },
    ymajorgrids,
    yminorgrids,
    minor tick num=1,
    minor grid style={draw=gray!25},
  },
}

\pgfplotsset{
  MicroPlot/.style={
    small,
    width = \linewidth + 15pt,
    height=0.8\linewidth,
    xlabel={Threads},
    ylabel={Throughput (M/s)},
    xlabel near ticks,
    x label style={at={(0.5,-0.12)},font=\small},
    ylabel near ticks,
    y label style={at={(-0.1,0.5)},font=\small},
    xmin = 0,
    xmax = 17,
    xtick = data,
    ymin = 0,
    scaled y ticks = base 10:-6,
    ytick style={draw=none},
    ytick scale label code/.code = {},
    yticklabel style={ /pgf/number format/fixed },
    ymajorgrids,
    yminorgrids,
    minor tick num=1,
    minor grid style={draw=gray!25},
  },
}

\pgfplotsset{
  YCSBPlot/.style={
    MicroPlot,
  },
}

\pgfplotsset{
  EvalLinePlot/.style={
    GenericLinePlot,
    height=0.61\columnwidth,
    xlabel={Load Factor},
    ylabel={Throughput (Millions/second)},
    xmin=0,
    xmax=100,
    minor x tick num=0,
    ymin=0,
    minor y tick num=1,
    grid=both,
  },
}

\pgfplotsset{
  PerfLinePlot/.style={
    EvalLinePlot,
    ymax=40,
  },
}

\pgfplotsset{
  PerfCacheLinePlot/.style={
    EvalLinePlot,
    ymax=80,
  },
}

\pgfplotsset{
  AggregatePlot/.style={
    footnotesize,
    ybar,
    bar width = 4.5,
    no marks,
    width=\columnwidth,
    height=0.7\columnwidth,
    xlabel near ticks,
    ylabel near ticks,
    ymin = 0,
    symbolic x coords = {insert, exist, false, remove},
    xticklabels = {{}, Insertion, Positive\\Lookup, Random\\Lookup, Deletion},
    xticklabel style = {align = center},
    major x tick style = transparent,
    enlarge x limits = 0.2,
    ymajorgrids,
    yminorgrids,
    minor y tick num = 1,
    fill,
    nodes near coords,
    every node near coord/.append style =
    {
      rotate = 90,
      anchor = west,
      fill = white,
      outer sep = 0,
      inner sep = 1,
      xshift = 1pt,
      /pgf/number format/fixed,
      /pgf/number format/precision = 1,
      font=\scriptsize
    }
  },
}

\pgfplotsset{
  RAMAggregatePlot/.style={
    AggregatePlot,
    ymax=50,
  },
  CacheAggregatePlot/.style={
    AggregatePlot,
    ymax=200,
  },
}

\pgfplotsset{
  IcebergStyle/.style = {color = cyan!75!blue, mark = *},
  DashStyle/.style    = {color = orange,       mark = square*},
  CuckooStyle/.style  = {color = red,          mark = diamond*},
  TBBStyle/.style     = {color = violet,       mark = triangle*},
  CLHTStyle/.style     = {color = green!75!black,       mark = pentagon*},
}

\begin{figure}[t]
  \centering
  \ref{intro-bar-legend}
  \begin{subfigure}{0.45\columnwidth}
  \begin{tikzpicture}
    \begin{axis}[
        IntroBarPlotWrite,
        legend entries = {IcebergHT, Dash, CLHT},
        legend columns = 3,
        legend cell align = {left},
        legend to name={intro-bar-legend}
      ]
      \addplot[IcebergStyle, fill]    table {data/iceberg_micro_pmem_bar_write.csv};
      \addplot[DashStyle, fill]    table {data/dash_micro_pmem_bar_write.csv};
      \addplot[CLHTStyle, fill]    table {data/clht_lb_res_micro_pmem_bar_write.csv};
    \end{axis}
  \end{tikzpicture}
\end{subfigure}
  \hspace{10pt}
\begin{subfigure}{0.45\columnwidth}
  \begin{tikzpicture}
    \begin{axis}[
        IntroBarPlotRead,
      ]
      \addplot[IcebergStyle, fill]    table {data/iceberg_micro_pmem_bar_read.csv};
      \addplot[DashStyle, fill]    table {data/dash_micro_pmem_bar_read.csv};
      \addplot[CLHTStyle, fill]    table {data/clht_lb_res_micro_pmem_bar_read.csv};
    \end{axis}
  \end{tikzpicture}
\end{subfigure}
  \label{fig:throughput-intro}
  \caption{Throughput for insertions, deletions, and queries (positive and
  negative) using 16 threads for PMEM hash tables. The throughput is
  computed by inserting $0.95N$ keys-value pairs where $N$ is the
  initial capacity of the hash table. (Throughput is Million ops/second)}
  \vspace{-1.5em}
\end{figure}
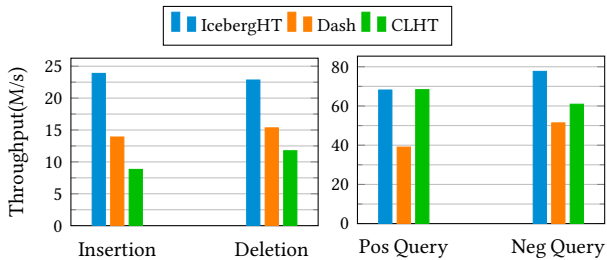

Hash tables are a core data structure in many applications, including
key-value stores, databases, and big-data-analysis engines, and are
included in most standard libraries.
Hash-table performance can be a substantial bottleneck for many applications~\cite{NealZu21,FanAn13,MetreveliZe12}.





With the advent of persistent-memory (PMEM) hardware, such as Intel Optane,
designing PMEM hash tables
has become an active field of research~\cite{MaierSaDe19, MaierSaDe16,
LuHaWa20,pfht,path-hashing,level-hashing,cceh,nvc-hashmap,LeeEtAl19-Recipe,HuCh21,phprx}.
Optane is cheaper than DRAM, enabling larger data
sets, but it is slower, with write bandwidth being roughy $2-5\times$ slower
than read bandwidth~\cite{pmem-measurements,YangKi20}.
Hash tables must be specifically designed for PMEM in order to achieve
both high performance and crash safety.

Despite several years of research on PMEM hash tables,
state-of-the-art PMEM hash tables---such as Dash~\cite{LuHaWa20} and
CLHT~\cite{david2015asynchronized}---utilize less than 35\%  of PMEM's raw
throughput for at least one of insertions and queries (see
\Cref{fig:throughput-dram-bar})
%
%
Furthermore, Dash and CLHT have space overheads of 2-3$\times$ (see
\Cref{tab:intro-space}).

In this paper, we introduce a new hash table, \textbf{\sysname}, that
is able to achieve over 60-70\% of the PMEM hardware throughput on both
insertions and queries, scales easily with additional threads, is
crash safe, and has space efficiency of over 85\% (i.e. space overhead
is less than $1/0.85\approx 17\%$).  \rev{R4D17}{\sysname also integrates easily
with existing PMEM software infrastructure  without requiring custom
allocators or PMDK implementations, whereas CLHT requires
custom support libraries.}





The design of PMEM (and other) hash tables typically involves
developing a hash table algorithm that minimizes read and write
amplification.\footnote{The \defn{write amplification} is the amount
  of data written to PMEM per insert/delete divided by the amount of data
  inserted/deleted.
  Similarly, the
   \defn{read amplification} is the amount of data read from PMEM per
   query divided by the amount of data output by the query.
}
In this paper, we argue that two stricter criteria,
\emph{referential stability} and \emph{low associativity} should be
optimized to yield high performance on PMEM.
%
%
As we will see, these two goals seem to be at odds with each other,
and part of the innovation  of our hash table design is that it
simultaneously achieves both.
Naturally, the third design goal for a high-performance hash table is \emph{compactness}, but compactness also seems at odds with referential stability and low associativity.

A hash table is said to be \defn{stable} if the position where an element is
stored is guaranteed not to change until either the element is deleted or the
table is resized~\cite{sandersstability,originalstability,KnuthVol3}.
Stability offers a number of desirable properties.  For example, stability
enables simpler concurrency-control mechanisms and thus reduces the performance
impact of locking.
Moreover, since elements are not moved, writing is
minimized, which improves PMEM performance.

The \defn{associativity} of a hash table is the number of locations where an
element is allowed to be stored.\footnote{Associativity is often associated with
caches that restrict the locations an item may be stored in.  Here we refer to
\emph{data structural associativity}, which is a restriction on how many locations a
data structure may choose from to put an item in, even on fully associative
hardware.} The best known low-associative (DRAM) hash table is the cuckoo hash
table~\cite{Pagh:CuckooHash,PaghRo01}.  In the original design, each element has
exactly two locations in the table where it is allowed to be stored, meaning
that the associativity is two.  Low associativity yields a different set of desirable
properties---most importantly, it helps search costs. For example, searching for an
element in a cuckoo hash table is fast because there are only two locations in
the table to check.  In addition, low associativity can enable us to further improve
query performance by keeping a small amount of metadata; see \Cref{sec:iceberght}.

In combination, stability can be used to achieve high insertion
throughput in PMEM, where writes are expensive, and low associativity can be use to achieve high query  performance.
Furthermore, we also show how stability enables locking and concurrency-control mechanisms to be simplified, leading to better multithreaded scaling and simpler designs for crash consistency.

Unfortunately, there is a tension between stability and low
associativity.  If a hash table has associativity $\alpha$, and
elements cannot move once they are inserted, then an unlucky choice of
$\alpha$ locations for $\alpha$ elements can block a $(\alpha+1)$st
element from being inserted.  As $\alpha$ decreases, the probability
of such an unlucky event increases.  Cuckoo hashing reduces the
probability of these bad events by giving up stability via
\defn{kickout chains}, which are chains of elements that displace each
other from one location to another. Practical
implementations~\cite{LiAn14} generally increase the number of
elements that can be stored in a given location---and thus the
associativity---to reduce the kickout-chain length and increase the
maximum-allowed \defn{load factor}, i.e, the ratio of the total number
of keys in the table to the overall capacity of the table.

Similarly, there is a three-way tension between space efficiency,
associativity, and stability. 
For example, cuckoo hash tables can be made stable if they are
overprovisioned so much that the kickout-chain length reaches 0.  Such
overprovisioning directly decreases space efficiency, but it also
increases associativity. 
Linear probing hash tables are stable (assuming they use tombstones to implement delete) but, as the
load factor approaches 1, the average probe length for queries goes
up, increasing associativity.  Other open-addressing hash tables
have a similar space/associativity trade-off.
Chaining hash tables are stable, but they have large associativity
and significant space overheads.  CLHT~\cite{david2015asynchronized} improves query performance despite
high associativity by storing multiple items in each node, but
this further reduces space efficiency.

\sysname is based on a new type of hash table, which we call \defn{iceberg hashing}.
Iceberg hash tables are the first to simultaneously achieve low associativity and stability, and they also have small space consumption.
To date, hash tables have had to choose between stability (e.g., chaining),
low associativity (e.g., cuckoo) or neither
(e.g., Robin Hood~\cite{CelisLaMu85,AmbleKn74}).
The techniques introduced in this paper also have ramifications to the
theoretical study of hash tables---we present a detailed study of
these implications, including closing various theoretical open
problems, in a companion manuscript~\cite{bender2021all}.  We
describe Iceberg hash tables in Section~\ref{sec:iceberght}.

\para{Results}
In this paper, we introduce Iceberg hashing and its implementation, \sysname.
We prove that Iceberg hashing simultaneously achieves stability and low associativity. 
Iceberg hashing is the first hash-table design to achieve both properties.
These guarantees give \sysname excellent performance on PMEM, as well
as on DRAM, and for workloads ranging from read-heavy to  write-heavy.
Specifically, \sysname accesses very few cache lines, both for queries and insertions, has low CPU cost, and has 
high load factor (and thus small space), with high probability.
Stability and low associativity enable simpler concurrency mechanisms, so that \sysname achieves nearly linear scaling with the number of threads, and it is crash safe on PMEM.






In building \sysname, based on our Iceberg hash table, we offer the following design contributions:
\begin{enumerate}
\itemsep0em 
\item We show how to achieve an efficient, practical metadata scheme
  that is adapted to practical hardware constraints.  The metadata
  scheme is small enough that the metadata for a bucket fits in a
  cache line, improving query performance and enabling nearly
  lock-free concurrency, i.e., locks are needed only for resizing.
\item A highly concurrent and thread-safe implementation of Iceberg
  hashing.  \sysname can scale almost linearly with increasing threads
  in PMEM, as well as DRAM, experiments.
\item Fenceless crash safety on PMEM.  Achieving crash-safety on PMEM
  often requires controlling the order in which cache lines get
  flushed to persistent storage, e.g., to ensure that an undo-log
  entry gets persisted before the changes to the hash table get
  persisted.  Since insertions in our hash table modify only a single
  cache line, we can achieve crash safety by simply persisting that
  cache line.
\item A simple, high-performance and concurrent technique for
  efficiently resizing Iceberg hash tables in a lazy online manner,
  thus reducing the worst-case latency of insertions.
\end{enumerate}

\para{Performance}
We evaluated \sysname on a system with Intel Optane memory.  We find that:
\begin{enumerate}[noitemsep]
  \item \textbf{Inserts and deletions:} \sysname insertions and deletions
    are roughly 50\% faster than Dash and $2-3\times$ faster than CLHT.
  \item \textbf{Queries:} \sysname positive queries are as fast as
    CLHT and negative queries are about 20\% faster.  \sysname queries
    are $1.5-2\times$ faster than Dash.
  \item \textbf{Space:}  \sysname achieves a space efficiency of 85\%, whereas
    Dash and CLHT have space efficiencies of 45\% and 33\%, respectively.
  \item \textbf{Scalability:} \sysname throughput scales nearly
    linearly in all our benchmarks.  CLHT also scales roughly linearly
    to 8 threads but scales slightly less efficiently than \sysname
    from 8 to 16 threads.  Dash hits a wall at 8 threads in several
    benchmarks.
  \item \textbf{YCSB:} \sysname is anywhere from $1\times$ to $8\times$
    faster than Dash and CLHT in our YCSB benchmarks.
\end{enumerate}

Although \sysname is designed for PMEM, we also compare it to
libcuckoo~\cite{LiAn14}, CLHT~\cite{david2015asynchronized}, and TBB~\cite{Pheatt08} in DRAM.  We find that \sysname outperforms these other
state-of-the-art DRAM hash tables on insertions and offers good but
not-quite-best query performance.  For example, \sysname is almost
twice as fast as the next fastest hash table on insertions in DRAM,
and it nearly matches the fastest hash table on positive queries, but
it is only about 75\% as fast as the fastest hash table on negative
queries and roughly 60\% as fast on deletions.  We believe this
insertion optimization at the cost of queries
reflects the fact that \sysname is optimized to minimize writes, which
are expensive in PMEM, resulting in stellar insertion performance.
Thus, \sysname is a strong choice for insertion-heavy workloads in DRAM.

\para{Roadmap}
In the rest of the paper, we discuss the various hash table designs
in and we give an overview of Iceberg
hashing and theoretical guarantees~\Cref{sec:iceberght}.
In~\Cref{sec:impl,resizing,multithreading,sec:pmem}, we present our
implementation of \sysname in DRAM and PMEM. \Cref{sec:eval} evaluates \sysname
and compares it with other hash tables.
We discuss related work in~\Cref{sec:related}.



\section{Iceberg Hashing} \label{sec:iceberght}
In this section, we begin by introducing Iceberg Hashing, a new,
stable, low-associativity hash-table design.  We then give the
theoretical basis for Iceberg hashing, proving the theorems that establish its correctness. 
%
In subsequent sections, we show how to exploit Iceberg hashing's low associativity to implement an efficient metadata scheme, explain how to make the hashtable concurrent, how to handle resizes, and how to ensure crash safety.  

\rev{R5O1}{The goal of this section is to establish the theoretical basis for the high performance we demonstrate in~\Cref{sec:eval}.  Of particular note for PMEM is that \sysname{} enables an unmanaged backyard that results in stability, which we will show is important for both high performance and crash safety on PMEM. These theoretical guarantees hold even in the presence of deletes. Previous hash-table designs have weak or no theoretical guarantees in the presence of deletes, e.g., cuckoo hashing. 
An important technical challenge is to guarantee stability and low associativity, which we simultaneously achieve in a hash table for the first time.}

\subsection{From Load Balancing to Iceberg hashing}

\rev{R4D7, R5O2}{
In this section, we have an overview of the design and design principles of
\sysname. 
\sysname{} is a three-level hash table, where most items are hashed into a very
efficient first level, some items are hashed into a less efficient second level,
and a few residual items are hashed into an overflow third level. \rev{R4D6}{The
first level is called the \defn{front yard} and the second and third levels are
called the \defn{backyard}.} In the remainder of the section, we describe how
each level is designed, and we give theorems to show that \sysname{} is correct
and fast.
Interestingly, the bounds in our main theorems are so tight that we are able to make all parameter choices in our implementation based on these theorems, as we describe below.

Consider a one-level hash table (which will correspond to the first level of \sysname{}).}
One way to design a hash table is to take an array and logically break
it into $m$ buckets of size $b$.  As items are inserted, they are
hashed to a random bucket and placed in any free spot of the bucket.
After inserting $n$ items, the expected number of items in each bucket
will be $h = n/m$ and the \defn{space efficiency} of the table will be
$bm/n = bm/hm = b/h$.  Thus, in order to optimize space efficiency, we want to minimize $b/h$.  \rev{R4D5}{But $b$ is a function of $h$, so the choices are not independent, as show in~\Cref{tab:maxfill}.  Note that in a balls-and-bins game, $b$ is the maximum fill of a bucket, because in our hash table, each bucket must be configured to be big enough to handle all insertions into that bucket.}

\begin{table}[]
\begin{tabular}{|c|c|c|}
\toprule
Ave fill $=h$          & Max fill $=b$          & Space Efficiency $= b/h$       \\
\midrule
$O(1)$  & $O(\log n / \log\log n)$ & $\Theta(\log n /  \log\log n) \gg \Theta(1)$ \\
\midrule
$\log n$     & $O(\log n)$            & $\Theta(1)$                         \\
\midrule
$\gg \log n$ & $h + O(\sqrt{h})$          & $1+o(1)$                       \\ 
\bottomrule
\end{tabular}
\caption{\boldmath The relationship between the average fill, $b$, and the
  maximum fill, $h$, in a balls-and-bins system is well
  understood~\cite{DBLP:conf/soda/BenderCCFJT19,DBLP:books/daglib/0012859}.}
  \label{tab:maxfill}
  \vspace{-2.5em}
\end{table}

The second observation is that, by using a backyard, we don't need to
get the number of overflows to 0.  Specifically, we configure the
front yard so that the number of overflows will be $O(n/\polylog{n})$.
Then we can use any hash table for the back yard as long as it has
$\Theta(1)$ space efficiency.  
In \cref{sec:theory}, we show that the overall space efficiency of
the hash table will be a remarkable $1+O(1/\log n)$.

\rev{R4D8}{We conclude that $h$ should be somewhat greater than $\log n$, so we set the bucket size to be 64.  This bucket size is bigger than a cache line but \sysname{} does not read the whole bucket.  Rather it will keep metadata to index the bucket.  \sysname{} finds itself in a sweet spot because, as will show below, buckets of size 64 are small enough that the metadata needed to index the items in a bucket fits in a cache line.

A smaller bucket size offers a poorer choice. 
Smaller buckets would either decrease the space efficiency, by increasing the number of buckets needed to prevent overflows, or increase the number of items that land in the less efficient backyard.  On the other hand, a larger bucket size does not decrease overflows but the metadata for bigger buckets no longer fits in a cache line.
}

\subsection{Bounding the Overflows}
\label{sec:theory}

In this section, we describe the theoretical basis for Iceberg hash
tables. The primary theorem we need is a bound on the number of items that will be placed in the backyard.




As before, we have  a hash table with a front yard consisting of an array broken into $m$ equal-size buckets.
Items are hashed to a single bucket and may be placed in any slot
in their bucket---if there is \emph{no} free slot in the bucket, then the item is placed in the \emph{backyard}. 
The hash table is stable: once inserted, items are not moved until they are deleted. 

The following theorem bounds the size of the backyard.



\rev{R4D4}{
\begin{theorem}\label{thm:ice}
  Consider a frontyard/backyard hash table that can hold up to $n$
  items.  Suppose further that the front yard consists of $m$ bins.
  When an item $x$ arrives, it is hashed uniformly into a bin $H(x)$.
  If bin $H(x)$ has room, the item is placed into the bin, and if bin
  $H(x)$ is full, it is placed into the backyard.  The capacity of a
  bin is determined by two parameters: $h \le n^{1/4} / \sqrt{\log n}$
  and $j \le \sqrt{h}$.  Specifically, each bin has capacity
  $h+ j\sqrt{h+1} + 1$.  Then at any moment over the course of
  $\poly{m}$ insertion/deletions where the table never has more than
  $n$ items, the number of balls in the backyard is
  $O(n / 2^{\Omega(j^2)} + n^{3/4}\sqrt{\log n})$ with probability
  $1 - 1 / \poly{n}$.
\end{theorem}}

\begin{proof}
For the sake of analysis, partition the bins into $K = \sqrt{n}$ collections $\mathcal{B}_1, \ldots, \mathcal{B}_{K}$ each of which contains $m / K$ bins. For each time $t$ and bin collection $\mathcal{B}_i$, define $R_{t, i}$ to be the set of balls $x$ that are present at time $t$ and satisfy $H(x) \in \mathcal{B}_i$. By a standard application of Chernoff bounds, we can deduce that, for any fixed $i, t$ we have 
\begin{align}
\begin{split}
|R_{i, t}| &\le n / K + O(\sqrt{(\log n) \cdot n / K})\\& \le \sqrt{n}
+ O(n^{1/4}\sqrt{\log n})
\end{split}
\label{eq:rit}
\end{align}
with high probability in $n$ (i.e., with probability $1 - 1 / \poly{n}$). Applying a union bound over all $i, t$, we find that \eqref{eq:rit} holds with high probability in $n$ for all $i, t$ simultaneously. Consider any possible outcome $R$ for the sets $\{R_{i, t}\}$, where the only requirement on $R$ is that \eqref{eq:rit} holds for all $i, t$;  we will show that if we condition on such an $R$ occurring, then the size of the backyard is $O(n / 2^{\Omega(j^2)} + n^{3/4}\sqrt{\log n})$ with high probability in $n$.

Consider some time $t$, and let $X_i$ be the number of balls that are in the backyard at time $t$ and that satisfy $H(x) \in \mathcal{B}_i$. Observe that the conditional variables $X_1 | R, X_2 | R, \ldots, X_{K} | R$ are independent (since $R$ fully determines which $x$ and $i$ satisfy $H(x) \in \mathcal{B}_i$). Thus, if we define
$$X | R := \sum_{i = 1}^{K} X_i | R,$$
then $X | R$ is a sum of independent random variables, each of which is (by \eqref{eq:rit}) deterministically in the range $[0, O(\sqrt{n})]$. We can therefore apply a Chernoff bound to $X | R$ to deduce that
$$\Pr[X | R \le \E[X | R] + O(\sqrt{Kn\log n})] \ge 1 - 1 / \poly{n},$$
Recalling that $K = \sqrt{n}$, we conclude that $X | R \le \E[X | R] + O(n^{3/4}\sqrt{\log n})$ with high probability in $n$. 

To complete the proof, it suffices to bound $\E[X | R]$ by $O(n / 2^{\Omega(j^2)})$. For this, in turn, it suffices to show that each ball $x$ present at time $t$ (there are up to $n$ such balls) satisfies
\begin{equation}
\Pr[x \text{ in backyard} \mid R] \le 1 / 2^{\Omega(j^2)}.
\label{eq:indball}
\end{equation}
To prove \eqref{eq:indball}, consider a ball $x$ that hashes to some collection $\mathcal{H}_i$. At the previous time $t_0 < t$ that $x$ was inserted, we have by \eqref{eq:rit} that there were at most $\sqrt{n} + O(n^{1/4} \sqrt{\log n})$ balls present that hashed to $\mathcal{H}_i$ (i.e., balls in the set $R_{i, t_0} \setminus \{x\}$); each of these balls has probability $K / m = K / (n / h) = Kh / n$ of hashing to the same bin as $x$, meaning that the number $Y$ of balls that hash to the same bin as $x$ at time $t_0$ satisfies
\begin{align*}\E[Y | R] &\le \frac{Kh\sqrt{n} + O(Khn^{1/4}\sqrt{\log
  n})}{n} \\&= h (1 + \sqrt{\log n} / n^{1/4}) \le h + 1.\end{align*}
The random variable $Y | R$ is just a sum of (up to) $n^{2/3} + O(n^{1/3} \sqrt{\log n})$ independent indicator random variables (one for each ball in $R_{i, t_0} \setminus \{x\}$). So by a Chernoff bound we have that
$$\Pr[Y \mid R \ge h + 1 + j \sqrt{h + 1}] \le 2^{-\Omega(j^2)}.$$ 
This implies \eqref{eq:indball}, which completes the proof. 
\end{proof}

The main consequence of this Theorem is that this simple bucketed
front-yard
design can hold all but $n/\poly{h}$ items, and by design the front yard is
also stable.  For example, if we set $h=\log n$ and $j=\Omega(\sqrt{\log\log n})$,
then $O(n/\log n)$ items will go to the backyard. \rev{R4D10}{The choice of $h=\log n$ suggests that the front-yard buckets should be of size 64, which we show in~\Cref{sec:eval} provides excellent performance.}

\subsection{The Backyard}

Iceberg hashing allows any of several backyard designs.  For \sysname,
we have selected a hash-table strategy based on the
power-of-2-choices.
We use power-of-two-choices in order to mitigating the space overhead of
the backyard.  The potential issue with using a power-of-two-choice hash table
is that queries and inserts level 2 must examine two buckets.  However, most
items reside in the front yard, so most queries need to examine only
the front yard, which means that the cost of checking two buckets in
level 2 will not substantially impact overall performance.

To analyze the space efficiency and overflow probability of the backyard,
let $z$ be the upper bound on the number of overflowing
items from Theorem~\ref{thm:ice}.
The backyard consists of an array of length
$\Theta(z\log\log{z})$, divided into $z$ buckets of size $\Theta(\log\log{z})$.
Items are hashed to two buckets and are placed into a slot in the bucket with
fewer items.

The following result of V\"{o}cking provides a theoretical guarantee
that the backyard will not overflow.

\begin{theorem}[\cite{vocking2003asymmetry}]
   Consider an infinite balls-and-bins process with $z$ bins in which at each
   step a ball is either inserted using the power-of-2-choices algorithm or an
   existing ball is removed, such that there are at most $hz$ balls present at
   any given step. Then the maximum load of any given bin is $(\ln\ln{z})/\ln{2}
   + O(h)$.
   \label{thm:vocking}
\end{theorem}

For level~2, the average bucket fill $h$ is less than 1, so
Theorem~\ref{thm:vocking} tells us that the number of items that overflow at
level 2 is quite small.  We store these items in a third level that
uses a standard chaining hash table.  So few items make it to the
third level that performance and space efficiency are negligible.
\rev{R4D10}{As noted above, \Cref{thm:vocking} suggests that level 2 buckets
should be of size $\ln\ln n$.  We use 8 as a coarse upper bound on log log n for
all practical purposes.}

\subsection{Summary}

\begin{figure}[t]
  \centering
  \resizebox{\columnwidth}{!}{%
  \includegraphics{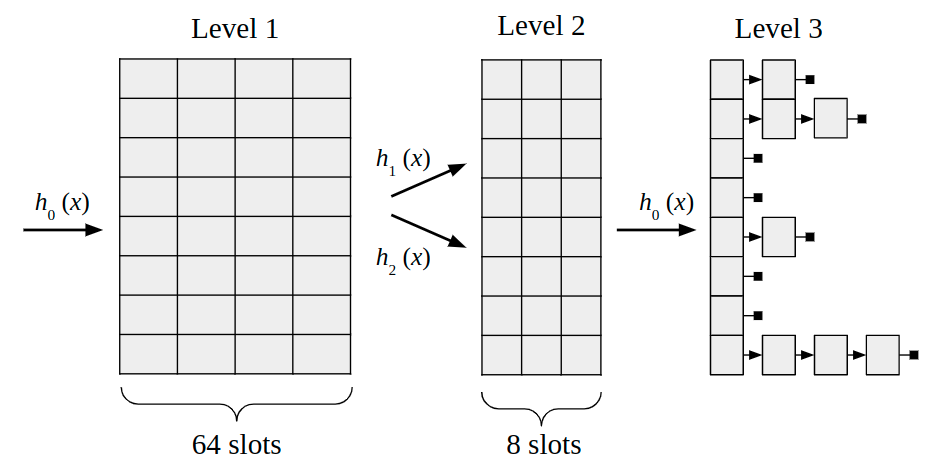}
}
  \caption{\boldmath Iceberg hash table block structure. Iceberg table has three levels.
  To insert a key value pair, we first hash the key $h_0(\mbox{key})$ and
  determine a block in level 1.
  If the block in level 1 is full, we try to insert it in level 2. In level 2,
  we hash the key twice $h_1(\mbox{key})$ and $h_2(\mbox{key})$ and insert the
  key in the emptier block.
  If the both blocks are full in level 2 then we insert the key value pair in
  level 3 block $h_0(\mbox{key})$.
  Level 3 contains a tiny fraction of keys (see~\Cref{tab:key-dist}) and choice
  of structure in level 3 does not have an impact on the hash table performance.
  \rev{R4D18}{}
  }
  \label{fig:iceberg_table}
  \vspace{-5pt}
\end{figure}

In summary, an Iceberg hash table consists of three levels, as shown in \Cref{fig:iceberg_table}.
Level 1 is a
power-of-one-choice front yard with buckets of size $\log n +
O(\sqrt{\log n\log\log n})$, level 2 is a power-of-two-choice table with
buckets of size $O(\log\log n)$, and level 3 consists
of a simple chaining hash table.

This design offers several benefits:
\begin{itemize}
  \item Such a table is stable: items never move after they are inserted.
  \item The number of buckets an item can reside in is only 4 (1
    bucket in level 1, 2 in level 2, and 1 in level 3).
  \item Most queries are satisfied by searching in level 1, so the
    average number of buckets accesses per query is just over 1.
  \item The buckets are small, so the associativity of the scheme is
    $\log n + \log \log n$ (plus level 3, which is rarely used).  So
    we can encode the exact slot of an element  using $O(\log\log
    n)$ bits. 
\end{itemize}

\rev{R4D3}{					
We conclude by noting that IcebergHT offers particular advantages on PMEM.  Specifically, it is stable and has low associativity and is backed by strong theoretical guarantees. This results in low read and write amplification, which are desired characteristics to achieve high performance on PMEM.  Furthermore, crash safety correctness follows almost directly from stability (see~\Cref{sec:pmem}).  The theoretical proofs and algorithmic novelty set up the PMEM-friendly design in the next section and are backed up by strong performance results in~\Cref{sec:eval}.}

\section{Implementation} \label{sec:impl}

We now describe how we implement metadata scheme and operations in \sysname.

\subsection{Metadata scheme} \label{sec:metadata}
This section describes our \rev{R4D9}{in-DRAM} metadata scheme that enables most queries
and inserts to complete by accessing only a single PMEM cache line.
Our goal is ambitious:  metadata is designed so that (1) metadata for each bucket fits on
a single cache line and (2) we can use vector instructions for all metadata operations.
\rev{R4D9}{Since metadata lives in DRAM, it costs substantially less to access
than PMEM.  In the event of a crash, we can recompute the metadata during
recovery, as explained in \Cref{sec:pmem}.}


One of the impediments to storing key-value pairs in large buckets as
in \sysname is that large buckets span multiple cache lines. This
hurts the cache efficiency, because operations may need to access
multiple cache lines per block.

\sysname addresses this concern by storing metadata for each block.
\rev{R2O2, R2O3, R4D10}{The metadata for a block of $k$ slots consists of an array of $k$
8-bit fingerprints, one per slot. If the slot holds a valid key, the
corresponding fingerprint is a hash of the key, otherwise the metadata
entry holds a special \texttt{EMPTY} fingerprint.  Note that we do not
reserve an entire bit to indicate empty/non-empty---we reserve a
single fingerprint value---so there are 255 valid fingerprints.}

The metadata scheme thus has a space overhead of 6.25\% for a 16 byte key-value
pair.  \rev{R4D10}{For smaller key-value pairs, the space overhead of the
metadata may be higher (e.g. 25\% for 4-byte key-value pairs) but, as we will
see in the evaluation section, many other
PMEM hash tables have much higher space overeheads.}
Importantly, because the blocks in level 1 have 64 slots and the blocks
in level 2 have 8 slots, the metadata for each block fits in a single cache
line.

During an insert operation, probing the metadata corresponding to a block
indicates which slots are empty in the block. The insert can then try to insert
the new key into one of those empty slots.

During a query operation, the fingerprint of the queried key can be checked
against the fingerprints in the metadata, yielding only those slots with a
matching fingerprint. This filters out empty slots as well as nearly all slots
with non-matching keys.

The metadata is also used to quickly compute the load in each block by counting
the number of occupied slots in the metadata block.

All of these operations can be implemented using vector instructions.
For example, to search for a fingerprint $x$ in a metadata vector $v$,
we use vector broadcast to construct a new vector $q$ where each entry
equals $x$ and then perform a vector comparison of $v$ and $q$.  To
find an empty slot, we do the same, except we set $x$ to
\texttt{EMPTY}.  To count the occupancy of a bucket, we perform the
search algorithm for \texttt{EMPTY}, which yields a bit-vector of
matching entries, and then use popcount to get the number of empty
slots.

\rev{R4D10}{Note that we do not use $2$ bits for EMPTY and RESERVED. Rather, these are two values out of $2^8$ $(256)$ values. Therefore, the fingerprints support $254$ values and the chance of collision is $64/254$.
}

\subsection{Operations}

\begin{algorithm}[t]
\centering
\caption{Insert (k, v)}
\label{alg:insert}
\begin{algorithmic}[1]
  \State $idx \gets h_0(k)$ \Comment{Compute the block index in level 1}
  \State $fp \gets \mathscr{F}(k)$  \Comment{Compute the fingerprint for key}
  \State \textsc{Lock}(lv1\_metadata[$idx$])
  \If {\textsc{ReplaceExisting}($k$, $v$)}
    \State \textsc{Unlock}(lv1\_metadata[$idx$])
    \State \Return \textsc{False}
  \EndIf
  \State $mask \gets \Call{Metadata\_Mask}{\mbox{lv1\_metadata}[idx], \texttt{EMPTY}}$
  \Comment{$mask$ is a bit-vector identifies empty slots in the block}
  \State $count \gets \Call{popcount}{mask}$  \Comment{Compute the number of empty
  slots}
  \If {$0 < \Call{popcount}{mask}$}
    \State $i \gets 0$
    \State $slot \gets \Call{Select}{mask, 0}$ \Comment{Compute the index of the first empty slot}
    \State $\mbox{lv1\_block}[idx][slot] \gets (k, v)$ \Comment{Store $(k, v)$ using 128-bit atomic store}
    \State $\mbox{lv1\_metadata}[idx][slot] \gets fp$
  \Else
    \State \Call{insert\_lv2}{$k,v,idx$} \Comment{Level 1 block is full. Try level 2}
  \EndIf
  \State \textsc{Unlock}(lv1\_metadata[$idx$])
  \State \Return \textsc{True}
\end{algorithmic}
\end{algorithm}

\begin{algorithm}[t]
\centering
\caption{Insert level2 (k, v)}
\label{alg:insertlv2}
\begin{algorithmic}[1]
  \Procedure{insert\_lv2}{$k, v, idx$}
    \State $idx1 \gets h_1(k)$ \Comment{Compute primary and secondary block
    indexes in level 2}
    \State $idx2 \gets h_2(k)$
    \State $fp1 \gets \mathscr{F}_1(k)$  \Comment{Compute primary and secondary
    fingerprints for the key}
    \State $fp2 \gets \mathscr{F}_2(k)$

    \State
    \State $mask1 \gets \Call{Metadata\_Mask}{\mbox{lv2\_metadata}[idx1], \texttt{EMPTY}}$
    \Comment{Compute a vector identifying empty slots in primary and
    secondary blocks}
    \State $mask2 \gets \Call{Metadata\_Mask}{\mbox{lv2\_metadata}[idx2], \texttt{EMPTY}}$
    \State $count1 \gets \Call{popcount}{mask1}$  \Comment{Compute the number of empty
    slots in primary and secondary blocks}
    \State $count2 \gets \Call{popcount}{mask2}$

    \State
    \If {$count2 < count1$}
    \State $idx1 \gets idx2$
    \State $fp1 \gets fp2$
    \State $mask1 \gets mask2$
    \State $count1 \gets count2$
    \EndIf

    \State $i \gets 0$
    \While{$i < count1$}
    \State $slot \gets \Call{Select}{mask1, i}$ \Comment{Compute the index of the
    next empty slot}
    \If {\Call{atomic\_cas}{\mbox{lv2\_metadata}[idx1][slot], \texttt{EMPTY}, $fp1$}}
    \Comment{Atomically set the metadata slot before updating the table}
    \State $\mbox{lv2\_block}[idx1][slot] \gets (k, v)$ \Comment{Store $(k, v)$ using 128-bit atomic store}
    \State \Return
    \EndIf
    \State $i \gets i+1$
    \EndWhile
    \State \Call{insert\_lv3}{$k,v,idx$} \Comment{Level 2 block is full.
    Try level 3}
  \EndProcedure
\end{algorithmic}
\end{algorithm}


Here we explain how to perform single-threaded operations in \sysname. Later
in~\Cref{multithreading}, we explain how to make these operations thread-safe.


\para{Inserts}
\rev{R2O1}{Our algorithm first searches whether $k$ already exists and, if so,
updates its associated value.}  For space, we omit the code for
replacing an existing item and show only the code for inserting a new
item.

We first try to insert the key-value pair in level 1. We hash the key using
$h_0$ to determine a block in level 1. If there is an empty slot in the block
then we insert the key-value pair and store the fingerprint in the
corresponding slot in the level 1 metadata. See the pseudocode
in~\Cref{alg:insert}.

If the block in level~1 is full, then we try to insert the key in level~2.  In
level~2, we use power-of-two-choice hashing to determine the block. We hash the
key twice and pick the emptier block. Similar to level~1, if there is an empty
slot in one of the blocks then we insert the key-value pair and store the
fingerprint in the corresponding slot in the level~2 metadata. See the
pseudocode in~\Cref{alg:insertlv2}.

Finally, if both the blocks in level~2 are full, then we insert the key in
level~3. We use the hash function from level~1 ($h_0$) to determine the linked
list to insert the key-value pair and insert at the head of the linked list. 
  
\para{Queries}
Similar to the insert operations, we perform queries starting from
level 1 and moving to levels 2 and 3 if we do not find the key in the
previous level.

During a query, we determine the block in a level in the same way as
we do during the insert. In level 1 and 3, there is only one block to
check and we use use hash function $h_0$ to determine the block.  In level 2,
the key can be present in either of the primary or the secondary block.
Therefore, we also perform a check in the secondary block if the key is not
found in the primary block.
 
Once we determine the block, we then perform a quick check to see if the
fingerprint of the queried key is present in the metadata of the block. Checking
the fingerprint requires a single memory access as all the fingerprints in a
given block fit inside a cache line.
If the fingerprint is not found in the metadata of the block then we can
terminate the query at that level and move to the next level. Otherwise, if one
or more fingerprint matches are found in the metadata of the block we then
perform a complete key match in the table for all possible matches and return a
pointer to the value if a key match is found.

If we are in level 3 during a query, we perform a linear search through the
linked list to find the key. However, buckets in level 3 are almost always empty
(<{}<1\% please refer to~\Cref{tab:key-dist}) and therefore we rarely have to
perform the linear search through the linked list.

\para{Deletions}
Deletions are performed similarly to queries. We first look
for the key starting from level 1 and then proceed to levels 2 and 3 if the key is not yet
found. Once the key is found, we first reset the
corresponding fingerprint in the metadata and then reset the key-value pair slot
in the table.

The pseudo-code for the query and remove operations follow the similar approach
as the insert operation pseudo-code. Therefore, they are omitted from the paper
to avoid redundancy.



\section{Resizing} \label{resizing}

This section describes how we resize the \sysname hash table when it reaches
full capacity.

The three levels of the \sysname hash table (see \Cref{sec:iceberght})
can be resized independently of each other.  We invoke a resize when
the load factor of the hash table reaches a predefined threshold,
which in \sysname has the default of 85\%.



In \sysname, we perform an in-place resize. In the in-place resize, we do not
allocate a separate table of twice the current size and move existing keys over
to the new table. Instead we use \emph{mremap}\footnote{mremap() expands (or
shrinks) an existing memory mapping~\cite{mremap}).} to remap the existing table space to twice
the size. To resize a given level, we first remap the level to twice the number
of current blocks. The size of each block remains the same (64 slots in level 1
and 8 slots in level 2) across resizes.
This means that during a resize, the space overhead of the table will be a most
$2\times$ instead of $3\times$ if we allocate a separate table of twice the size. 

Doing in-place resize means that only about half the existing keys (rather than
all) need to be moved to a new location because each item $x$'s bucket is
computed as $h(x) \bmod m$, where $m$ is the number of buckets in the table.
We move each key-value pair by first inserting
it into its new block (in the same level) and then deleting it from
its old block.

\rev{R5O5}{
The shrink can be performed in the similar way as the doubling. The keys from
the second half of the table can be moved to the first half by rehashing the keys.
Once the move is complete, the second half of the table can be freed.}

\subsection{Guaranteeing Balanced Levels After Resizing}

In this subsection, we argue that, as the table is dynamically resized, the
bounds from Section \ref{sec:theory} 
on the number of elements that overflow from levels 1 and 2 continue to hold.
The bound on the number of overflow elements from level 1 follows from
essentially
the same argument as in Theorem \ref{thm:ice}, so we will focus here on showing
that the bins in level 2 remain balanced.

Whenever the size of level 2 doubles, from $m$ bins to $2m$ bins, each bin $i$
can be thought of as splitting into two bins $i$ and $m + i$; each of the
elements that were in bin $i$ move to bin $m + i$ with probability $50\%$
(depending on the element's hash). We are not aware of any past bounds for the
maximum fill of a bin when bins are split in two from time to time.  Here, we
provide a lemma showing that the nice load-balancing property of
power-of-2-choice bin selection (i.e., Theorem \ref{thm:vocking}) is maintained, even when using our
resizing scheme. The proof can be viewed as an extension of the witness-tree
techniques used in \cite{vocking2003asymmetry}.

\begin{lemma}
Start with $M_0$ empty bins, and perform $N \le \poly{M_0}$ ball insertions. Double the bins whenever the current number $n$ of balls in the system surpasses $m / 4$, where $m$ is the current number of bins. At any given moment, the number of balls in the fullest bin is guaranteed to be $O(\log \log N)$ with probability $1 - 1 / \poly{N}$.
\end{lemma}
\begin{proof}[Proof Sketch]
For each ball $u$, define $n_u$ (resp. $m_u$) to be the number of balls (resp. bins) that were present when $u$ was inserted. As an invariant, we always have $n_u \le m_u / 4$. 

If a given ball $x$ has height $\Theta(\log \log N)$ then we can construct a depth-$\Theta(\log \log N)$ \defn{witness tree} $T$ of balls, where $x$ is the root, and where the children, $v_1$ and $v_2$, of any given node $u$ are determined as follows: if $u$ was placed at height $\ell$ when it was inserted, then $v_1, v_2$ are the balls that were at height $\ell - 1$ in bins $h_1(u, m_u)$ and $h_2(u, m_u)$.

We claim that, for any given ball $u$, if $u$ were to be a node in $T$, then the expected number of ways that we could hope to assign children to $u$ is at most $1/4$. Indeed, there are $\binom{n_u}{2} \approx n_u^2 / 2$ ways to choose two nodes $v_1, v_2$ that were present when $u$ was inserted, and the probability that both $v \in \{v_1, v_2\}$ satisfy $\{h_1(v, m_u), h_2(v, m_u)\} \cap \{h_1(u, m_u), h_2(u, m_u)\} \neq \emptyset$ is at most $\frac{4}{m_u^2}$. So the expected number of ways that we can assign children to $u$ is at most
$$\frac{n_u^2}{2} \cdot \frac{4}{m_u^2} = \left(\frac{2n_u}{m_u}\right)^2 \le \left(\frac{1}{2}\right)^2 = \frac{1}{4}.$$

Assume for simplicity that all $\polylog{n}$ of $T$'s nodes are distinct balls.\footnote{Formally, we can reduce to this case via standard pruning arguments, as in, e.g., \cite{vocking2003asymmetry}.} We have shown that, for each ball $u$, the expected number of ways that we can assign children to $u$ is $1/4$. Using this, one can argue that the expected number of valid configurations for the full tree $T$ with $\polylog{N}$ parent/child relationships is at most $1/ 4^{\polylog{N}} \le 1 / \poly{N}$. The probability of such a $T$ existing is therefore at most $1 / \poly{N}$.
\end{proof}


\section{Multi-threading}\label{multithreading}

We now describe how we implement thread-safe operations in \sysname. We
first describe how to synchronize among threads performing insert, query, and
delete operations. Afterwards, we explain how to synchronize among threads
when a level resizes.

\subsection{Thread-safety across operations} \label{thread-operation}

\rev{R4D14}{
We use one bit in the level 1 metadata as a lock.
For level 1, the metadata
consists of an array of 64 8-bit fingerprints.  We steal one bit from
one of the fingerprints to serve as the lock bit.  Consequently, that
fingerprint slot is only 7 bits and has a slightly higher
false-positive rate.

When a thread wants to insert a key that hashes to block $i$ in level
1, it first sets the lock bit for block $i$ using an atomic
fetch-and-or loop.  It holds this lock for the entire duration of the
insert, i.e. even if the element ends up inserted in level 2 or 3.
This ensures that inserts/updates/deletes of the same key cannot
execute concurrently, since they will both attempt to acquire the same
lock.

After acquiring the lock, the thread checks whether the key already
exists in any level and updates or deletes it, depending on the
requested operation.

When inserting a key that does not already exist in the hash table, we
first check for an empty slot in level 1 by using the metadata.  If we
find one, then we use a 128-bit atomic write to store the key and
value in the slot and update the fingerprint in the metadata.  Since
we hold a lock on the level 1 block, no additional synchronization is
necessary.

If the insertion goes to level 2 or 3, then we need to carefully
update the bucket and metadata because the locks on level 1 do not preclude
other threads operating
on the same level 2 or 3 bucket (but not the same key).  In level 2,
we find a metadata slot holding \texttt{EMPTY}, CAS our fingerprint into the
metadata slot, claiming it for our operation, and then write the
key-value pair into the slot using a 128-bit atomic write.
In level 3, we use an array of 1-byte integers to lock the linked list in which
we want to insert the key. We acquire a lock on the linked list using an atomic
\emph{test-and-set} instruction.

To support concurrent deletes and queries, we reserve a special
``invalid'' key.  Deletes reset the slot to the invalid key and then
set the corresponding fingerprint to \texttt{EMPTY}.
\rev{R4D15}{Note that we
can still allow the application to insert a key that is equal to our
special ``invalid'' key.  We just need to set aside a special location
for storing the associated value and a bit indicating whether the key
is present or not.  Concurrent updates can be made safe by using
\texttt{cmpxchg16b} to update the associated value and
the ``present'' bit atomically.
Queries must also special-case this key to check
the designated location instead of performing the standard lookup
algorithm.  They must also use 128-bit loads to get
the ``present'' bit and the value in one atomic read.}

Queries are lockless on levels 1 and 2.  They proceed through the levels, examining any
slots with a matching fingerprint.  They load the key-value pair from
a candidate slot using 128-bit atomic reads and then check whether the
key read from the slot is valid and actually matches the queried keys.
On level 3 they check for bucket emptiness locklessly but acquire locks
on buckets before searching in them.
Since all slots are read and written using 128-bit atomic operations,
and since buckets on level 3 are locked,
queries are guaranteed to see only entries with either invalid keys (which are ignored) or
with correct key-value pairs.
}

\subsection{Multi-threaded performance analysis}

Each insert, delete, or query dirties exactly one PMEM cache line,
i.e.  for the slot affected by the operation.  As for the metadata,
each mutation also dirties the level 1 metadata cache line (in DRAM)
for the target key's block (to acquire the lock).  If the insert does
not go into level 1, then it will also access 2 metadata cache lines
for level 2, and will dirty one of them.  Level 3 is so rarely used
that we can largely ignore it.  As our evaluation shows, over 90\% of
the keys go in level 1, so the average number of DRAM cache lines
accessed is around 1.2, and the average number dirtied is around 1.1.

Furthermore, since the cache line accesses are determined by the hash
of the key, they are independent (unless there are some hot keys that
get frequently updated) and therefore it is unlikely that two threads
will attempt to access/dirty the same cache lines at the same time.
\rev{R2O1}{Hot keys that are frequently updated are a genuine scaling
  bottleneck for almost all hash tables, including \sysname.}

Queries are invisible, i.e. they are lock free and dirty no cache
lines.




\subsection{Thread-safety across resizes}

\para{Initiating resizes}
When a resize is invoked, the table structure goes through the
memory-doubling phase, which requires a global lock on the hash table. During
the doubling phase, the insert, query, and delete operations cannot
operate on the table.  Thus, the table has a global reader-writer lock
for synchronizing between the memory-doubling step and all other
operations.  All other operations grab the global lock in read-mode, a
thread performing the memory-doubling step grabs it in write mode.

\rev{R4D19}{ The global lock is implemented as a distributed readers-writer
lock~\cite{DBLP:conf/spaa/LevLO09} \alex{suggest kill:} so that threads
acquiring the lock in read mode do not thrash on the cache line containing the
count of the number of readers holding the lock.} 


Each insertion checks the current load factor of the hash table and
performs a memory-doubling step if the load factor is above a
configurable threshold.  In order to ensure high concurrency,
insertion threads first check the load factor while holding the global
lock in read mode.  If a thread detects that a resize is needed, it
releases the global lock in read mode, reaquires it in write mode,
and rechecks the load factor.  If it is still above threshold, then
it performs the memory-doubling step, releases the global lock, and then
performs an insertion, as described below.

\alex{Should this be blue text?}
Recall that we ensure there is at most one operation per key by
locking the level 1 block for a key being inserted, updated, or
deleted.  A memory-doubling step changes the mapping from keys to
level 1 blocks, and hence changes the lock for each key.  We need to
ensure that there are not two threads operating concurrently but using
different key-to-lock mappings.  The global resizing lock solves this
problem by waiting for all in-flight mutations to complete before
beginning the resize.  Thus, during the resize, there are no threads
holding any locks on level 1 blocks.  After the resize completes,
mutations can resume, using the new key-to-lock mapping.

\rev{R4D11, R4D12, R4D13}{
\para{Concurrency of block moves and other operations} After the
memory-doubling step, existing key-value pairs must be moved to their
new location in the table.

We refer to blocks in the first half of the table as \defn{old blocks}
and blocks in the second half of the table as \defn{new blocks}.  Each
new block has a corresponding old block.

One clearly safe way to perform this step is to freeze the world,
perform all the moves, and then let other operations proceed.  Rather
than freezing the world, we simulate this by moving blocks the first
time any insert, update, or delete operation attempts to access them.
Concretely, during a resize, we maintain an additional \defn{moved}
flag for each old block.  The flag can be in one of three states:
\texttt{UNMOVED}, \texttt{IN-FLIGHT}, or \texttt{MOVED}.  Initially
all old blocks are marked as \texttt{UNMOVED}.  Whenever an insert,
update, or delete is about to access a block, it first checks the
state of the corresponding old block.  If the old block is in the
\texttt{UNMOVED} state, then the thread attempts to CAS the block's
state to \texttt{IN-FLIGHT}.  If the CAS fails, then the thread waits
until the state is \texttt{MOVED}.  If the CAS succeeds, then the
thread iterates over the block, moving key-value pairs to their new
block.  The thread then sets the block's state to \texttt{MOVED}.  The
operation can then continue its execution.

Queries do not check the moved flags, so we need to ensure that
queries and concurrent moves will not result in incorrect answers.
Queries check both the old and new locations for a key, in that order.
Moves ensure that each key-value pair is written to its new location
before erasing it from its old location.  Thus queries will never miss
an item in the table.

As an optimization, we also maintain a counter of the number of blocks
that still need to be moved.  Threads check this counter after
acquiring the global resize lock in read mode.  If the counter is 0,
then threads can skip the above additional work.  Thus, in the common
case when there is no on-going resize, operations do not incur the
overhead of checking moved flags or additional locations for a key.
Furthermore, since the count of blocks to be moved is never modified
when a resize is not in progress, each core can keep this counter in
its local cache, making the counter check very cheap.

}


\section{Crash Consistency and Performance on PMEM}
\label{sec:pmem}

\para{Crash safety}
Because \sysname{} is stable, crash consistency is straightforward.

Because all the data in levels 1 and 2 is accessed by computing an offset using
block numbers, there are no direct pointers into them, and so there is no need
for additional pointer swizzling. The linked lists in level 3 allocate nodes by
offset from a fixed array, which is mapped into PMEM. These offsets are then
used to reference the nodes.

Recall that all metadata is kept in volatile memory, so that only the data is
kept in PMEM. This data is stored in \rev{R2O4}{several large preallocated
sparse files on a PMEM-backed DAX file system, one each for levels 1 and 2, and
2 for level 3 (one for the linked list heads and one for allocating nodes).} A
specially designated value is used to indicate if a key or value is invalid,
and the key-value pair is considered invalid (and therefore free) if either key
or value is invalid.

An insert or deletion is persisted by writing the item into a slot (residing in
a block in a level on PMEM), and then performing a cache line writeback
instruction followed by an sfence, using PMDK~\cite{pmdk}. One small issue is
that persistent memory guarantees atomicity only for 8-byte stores, but we must
write 16 bytes to insert a key-value pair. However, because the key-value pair
is considered invalid if either key or value is invalid, we can store them in a
slot in either order, or the stores can even be reordered by the CPU, and the
hash table will always be in a consistent state. This eliminates the need for a
fence between storing the value and storing the key.

\rev{R2O4}{A global metadata file is used to store the initial size of the
array as well as the number of (doubling) resizes that have been performed.
Note that this file is only modified when a resize is initiated.}
Resizes first initialize the new PMEM data region to consist of invalid
key-value pairs, then updates and persists the table size in the global
metadata file, before updating the size in volatile memory.

Recovery consists of reading through the data array and rebuilding the metadata
for each valid key-value pair found. Because data from an in-progress resize
may not have been moved, recovery must check that each key-value pair is in the
correct block, and move it if it is not. Because this can be performed using a
sequential scan, the process is efficient.

\rev{R2O4, R4D1}{For example, consider a table initialized with $2^{24}=16777216$
level 1 slots (18874368 slots total in levels 1 and 2), into which is inserted
$2^{26} * 1.07 \approx 71.8\textrm{M}$ items, which causes 2 resizes, after
which the table is dismounted or crashes (dismount only performs deallocation).
Recovery on a single thread then takes 0.48\,s, recovering 173\,M slots per
second and 148\,M items per second (roughly 63$\times$ faster than individual
insertions). Furthermore this process is easily parallelized.}

\rev{R4D3}{
\para{Performance}
Changes to the hash table (i.e. inserts, deletes, and updates), modify
a single PMEM cache line unless they go to level 3, which we show in
our experiments is extremely rare.  Positive queries almost always
access a single PMEM cache line, plus occasional additional cache
lines from false positives in the metadata.  Negative queries also
almost always touch only a single PMEM cache line to examine the head
of the queried key's bucket in level 3 (plus, like other queries, any
false positives from the metadata checks in level 1 and 2).  We could
eliminate even that PMEM access by maintaining in-DRAM metadata about
the emptiness of each bucket in level 3, but we have not found it
necessary to do so.  Since inserts, deletes, and updates must query
for the target key, they may also occasionally access (but not modify)
extra PMEM cache lines due to metadata false positives.

So, in summary, all operations access a single PMEM cache line in the
common case.
}


\section{Experiments}\label{sec:eval}

In this section, we evaluate the performance of \sysname hash table. We compare
\sysname against two state-of-the-art concurrent PMEM hash tables,
Dash~\cite{LuHaWa20} and CLHT~\cite{david2015asynchronized} from
the RECIPE library~\cite{LeeEtAl19-Recipe}.
In our evaluation, we have used the Dash-Extendible Hashing (Dash-EH) variant
from the Dash-enabled hash tables. Dash-EH offers faster performance compared to
other Dash variants. For CLHT, we have used the CLHT\_LB\_RES variant which is
lock-based and supports resizing. The CLHT\_LB\_RES variant is ported to PMEM in
the RECIPE library~\cite{LeeEtAl19-Recipe}.

While \sysname primarily targets PMEM, its design also yields strong DRAM
performance. Therefore, we additionally evaluate \sysname on DRAM. On DRAM, we
compare \sysname against state-of-the-art concurrent in-memory hash tables,
libcuckoo~\cite{LiAn14}, Intel's threading building blocks (TBB) hash
table~\cite{Pheatt08}, and CLHT~\cite{david2015asynchronized}. Similar to the
PMEM evaluation, we use CLHT\_LB\_RES variant of CLHT.

We evaluate hash table performance on three fundamental operations: insertions,
lookups, and deletions. We evaluate lookups both for keys that are present and
for keys that are not present in the hash table. We also evaluate these hash
tables on multiple application workloads from YCSB~\cite{CooperSiTa10}, as well
as for space efficiency and scalability.
\rev{R2O2, R4D23} {
In \sysname, we use MurmurHash to compute the $h_0$, $h_1$, and $h_2$.}

The goal of this section is to answer the following questions:

\begin{enumerate}
\item How does \sysname performance compare to other hash tables when hash
   tables are on PMEM?
\item How does \sysname scale with increasing number of threads compared to
  other hash tables?
\item How does \sysname compare to other hash tables in terms of space
  efficiency and instantaneous throughput?
\item What is the impact of hash table resizing on the latency of operations in
  \sysname?
\item How does \sysname compare to libcuckoo, TBB, and CLHT when hash tables are
  in DRAM?
\end{enumerate}

\subsection{Other hash tables}\label{sec:other-ht}
%

CLHT~\cite{david2015asynchronized} and TBB~\cite{Pheatt08} are both
chaining-based hash tables. They use a linked list to handle collisions. They
dynamically allocate a new a node and add it to the linked list to insert a key
if the head bin is already occupied. Their space usage is also suboptimal
compared to other hash table designs.
Dash~\cite{LuHaWa20} is based on extendible hashing~\cite{fagin1979extendible}.
A directory is used to index (or store pointers to) the blocks that store
key-value pairs. Similar to chaining-based hash tables, Dash also perform
dynamic allocation of nodes at run time to add new keys.
In cuckoo hash table~\cite{LiAn14},  a pre-allocated array of blocks is
maintained where each block can store up to four key-value pairs. Cuckoo
hashing~\cite{Pagh:CuckooHash,PaghRo01} is used to perform insertions. Unlike
chaining-based or extendible hashing, there is not dynamic allocation in cuckoo
hash table.


\begin{figure*}
   \centering
   \ref{pmem-legend}\\
   \begin{subfigure}{0.25\linewidth}
      \begin{tikzpicture}
         \begin{axis}[
               MicroPlot,
               ymax = 30000000,
               legend entries={\color{black}\sysname, \color{black}\dash, \color{black}CLHT},
               legend columns=3,
               legend to name={pmem-legend}
            ]
            \addplot[IcebergStyle] table[x=threads, y=insert] {data/iceberg_micro_pmem_iine_blocklock.csv};
            \addplot[DashStyle]    table[x=threads, y=insert] {data/dash_micro_pmem.csv};
            \addplot[CLHTStyle]    table[x=threads, y=insert] {data/clht_lb_res_micro_pmem.csv};
         \end{axis}
      \end{tikzpicture}
      \caption{Insertion}
      \label{pmem-insertion}
   \end{subfigure}%
   \begin{subfigure}{0.25\linewidth}
      \begin{tikzpicture}
         \begin{axis}[
               MicroPlot,
               ymax = 120000000,
            ]
            \addplot[IcebergStyle] table[x=threads, y=positive] {data/iceberg_micro_pmem_iine_blocklock.csv};
            \addplot[DashStyle]    table[x=threads, y=positive] {data/dash_micro_pmem.csv};
            \addplot[CLHTStyle]    table[x=threads, y=positive] {data/clht_lb_res_micro_pmem.csv};
         \end{axis}
      \end{tikzpicture}
      \caption{Positive Query}
      \label{pmem-positive}
   \end{subfigure}%
   \begin{subfigure}{0.25\linewidth}
      \begin{tikzpicture}
         \begin{axis}[
               MicroPlot,
               ymax = 120000000,
            ]
            \addplot[IcebergStyle] table[x=threads, y=negative] {data/iceberg_micro_pmem_iine_blocklock.csv};
            \addplot[DashStyle]    table[x=threads, y=negative] {data/dash_micro_pmem.csv};
            \addplot[CLHTStyle]    table[x=threads, y=negative] {data/clht_lb_res_micro_pmem.csv};
         \end{axis}
      \end{tikzpicture}
      \caption{Negative Query}
      \label{pmem-negative}
   \end{subfigure}%
   \begin{subfigure}{0.25\linewidth}
      \begin{tikzpicture}
         \begin{axis}[
               MicroPlot,
               ymax = 30000000,
            ]
            \addplot[IcebergStyle] table[x=threads, y=remove] {data/iceberg_micro_pmem_iine_blocklock.csv};
            \addplot[DashStyle]    table[x=threads, y=remove] {data/dash_micro_pmem.csv};
            \addplot[CLHTStyle]    table[x=threads, y=remove] {data/clht_lb_res_micro_pmem.csv};
         \end{axis}
      \end{tikzpicture}
      \caption{Deletion}
      \label{pmem-deletion}
   \end{subfigure}%
   \\
  \begin{subfigure}{0.25\linewidth}
      \begin{tikzpicture}
         \begin{axis}[
               YCSBPlot,
               ymax = 20000000,
            ]
            \addplot[IcebergStyle] table[x=threads, y=load] {data/iceberg_ycsb_pmem_iine_blocklock.csv};
            \addplot[DashStyle]    table[x=threads, y=load] {data/dash_ycsb_pmem.csv};
            \addplot[CLHTStyle]    table[x=threads, y=load] {data/clht_lb_res_ycsb_pmem.csv};
         \end{axis}
      \end{tikzpicture}
      \caption{YCSB Load}
      \label{pmem-load}
   \end{subfigure}%
   \begin{subfigure}{0.25\linewidth}
      \begin{tikzpicture}
         \begin{axis}[
               YCSBPlot,
               ymax = 30000000,
            ]
            \addplot[IcebergStyle] table[x=threads, y=a] {data/iceberg_ycsb_pmem_iine_blocklock.csv};
            \addplot[DashStyle]    table[x=threads, y=a] {data/dash_ycsb_pmem.csv};
            \addplot[CLHTStyle]    table[x=threads, y=a] {data/clht_lb_res_ycsb_pmem.csv};
         \end{axis}
      \end{tikzpicture}
      \caption{YCSB Run A}
      \label{pmem-runa}
   \end{subfigure}%
   \begin{subfigure}{0.25\linewidth}
      \begin{tikzpicture}
         \begin{axis}[
               YCSBPlot,
               ymax = 60000000,
            ]
            \addplot[IcebergStyle] table[x=threads, y=b] {data/iceberg_ycsb_pmem_iine_blocklock.csv};
            \addplot[DashStyle]    table[x=threads, y=b] {data/dash_ycsb_pmem.csv};
            \addplot[CLHTStyle]    table[x=threads, y=b] {data/clht_lb_res_ycsb_pmem.csv};
         \end{axis}
      \end{tikzpicture}
      \caption{YCSB Run B}
      \label{pmem-runb}
   \end{subfigure}%
   \begin{subfigure}{0.25\linewidth}
      \begin{tikzpicture}
         \begin{axis}[
               YCSBPlot,
               ymax = 80000000,
            ]
            \addplot[IcebergStyle] table[x=threads, y=c] {data/iceberg_ycsb_pmem_iine_blocklock.csv};
            \addplot[DashStyle]    table[x=threads, y=c] {data/dash_ycsb_pmem.csv};
            \addplot[CLHTStyle]    table[x=threads, y=c] {data/clht_lb_res_ycsb_pmem.csv};
         \end{axis}
      \end{tikzpicture}
      \caption{YCSB Run C}
      \label{pmem-runc}
   \end{subfigure}%

\caption{Performance of hash tables on PMEM on micro and YCSB workloads.
(Throughput is Million ops/second)}
\label{fig:performance-pmem}
\end{figure*}
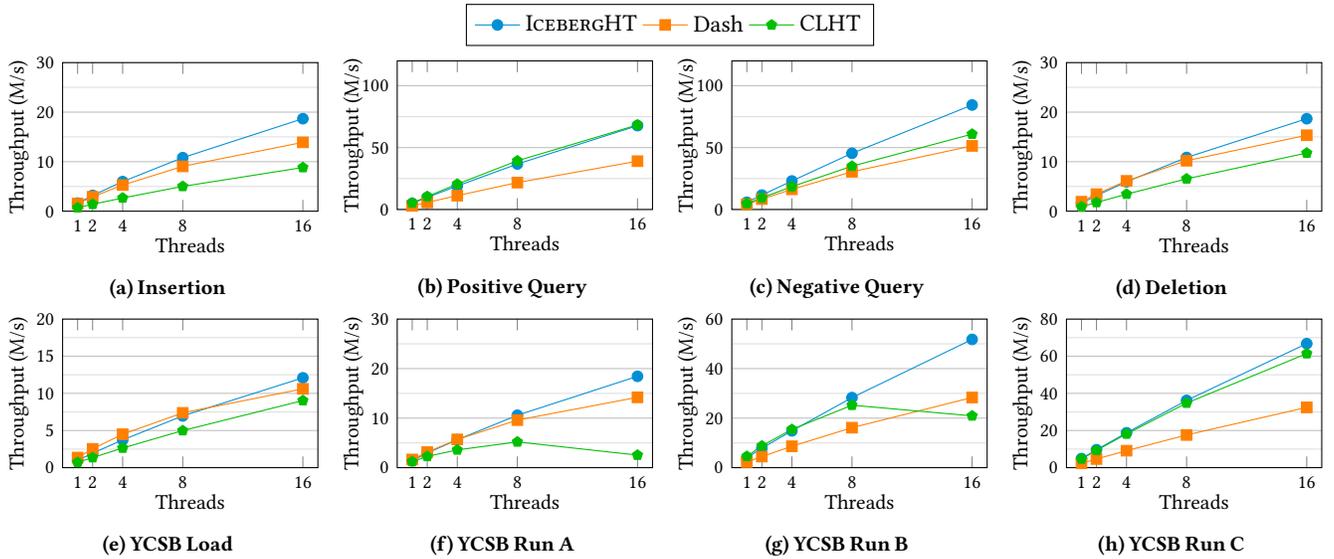

\subsection{Experimental setup}

In our evaluation, we perform two sets of benchmarks: micro benchmarks and
application workloads. For both types of benchmarks, we evaluate the scalability
of hash table operations with increasing number of threads.
%

\para{Microbenchmarks}
We measure performance on insertions, deletions, and lookups which are
performed as follows.  We generate 64-bit keys and 64-bit values from
a uniform-random distribution to be inserted, removed or queried in
the hash table.  We configured each hash table to have as close to
$2^{26}$ slots as possible, and we filled each hash table to its
maximum recommended load factor.  Specifically, we configured CLHT to
use $2^{25}$ buckets, each with 3 slots\footnote{We also tried configured CLHT with $2^{26}/3$ slots, but its performance is much worse when the number of slots is not a power of $2$.}.  Dash and TBB were
initialized with a target size of $2^{26}$, libcuckoo was initialized
with $2^{26}$ slots, and \sysname was initialized with a front yard of
$2^{26}$ slots, for a total of $(1 + 1/8)2^{26}$ slots, when also
counting level 2.
We then inserted $0.95 N$ keys into each hash table, where $N$ is the
number of slots in the table (e.g. $N=3\times 2^{25}$ for CLHT, $(1+1/8)2^{26}$ for \sysname, and $2^{26}$ for all other hash tables).
We report the
aggregate throughput going from empty to 95\% full as the insertion throughput.

Once the data structure is 95\% full, we perform queries for keys that exist and
keys that do not exist in the hash table to measure the query throughput for
both positive and negative queries. For positive lookups, we query keys that are
already inserted and for negative lookups we generate a different set of 64-bit
keys than the set used for insertion. The negative lookup set contains almost
entirely non-existent keys because the key space is much bigger than the number
of keys in the insertion set. Empirically, 99.9989\% of keys in the negative
lookup query set were non-existent in the input set.
%
%
We then remove a random selection of existing keys from the hash table until
its load factor reaches $\approx50\%$ and report the aggregate deletion
throughput.


In order to isolate the performance differences between the hash tables, we do
not count the time required to generate the random inputs to the hash tables.
%

\para{Application workloads}
We also measure the hash table performance on YCSB~\cite{CooperSiTa10}
workloads. We use YCSB workloads A, B, and C in our evaluation. Workload A has a
mix of 50/50 reads and writes. Workload B has a 95/5 reads/write mix. Workload C
is 100\% read. We do not include other YCSB workloads as operations required by
other workloads are not supported by these hash tables.
The YCSB workloads consist of a load and a run phase. In the load phase, we
insert 64M keys and values (64-bit keys and 64-bit values same as in the micro
benchmark) generated using a uniform random distribution. The load phase
configuration is the same for all three workloads. The keys are generated using the
YCSB workload generator. All the hash tables are configured as in the microbenchmarks, except we target $2^{24}
\approx 17M$ slots instead of $2^{26}$.  This ensures that they resize twice during the load phase of $64$M keys.
In the run phase, we perform a mixed workload depending upon the workload type.
In order to make the performance in the run phase a representative of the actual
performance of the hash tables, we make sure that the run phase is large enough
so that the table doubles its size. Doing this enables us to include the impact
of a resize on the insert and query operations in the hash table and ensures
that resizes do not unfairly bias the benchmarks. 

We achieve this by keeping the number of keys inserted in the run phase the same as
the number of keys that are present in the hash table at the start of the run
phase.
Therefore, the run phase in workload A consists of 128M operations out of which
64M (50/50 reads and writes) are inserts. Similarly, the run phase in workload B
consists of 1.28B operations out of which 64M are inserts (95/5 reads/write
mix). Workload C does not have any inserts and only contains 64M read
operations.

\para{Speed/space tradeoff}
To measure how different hash tables can trade space efficiency for speed,
we fill the hash table from empty to 95\% full in increments of 5\%.
Data items are generated as in the microbenchmarks.
We record the throughput and max RSS (resident-set size) in each increment.
To report the memory usage of the hash table we subtract the total memory
allocated by the driver process from the Max RSS reported by \emph{getrusage}.

To measure the space usage of PMEM hash tables, we measure the size of the file
created by the hash tables on PMEM. In \sysname, the PMEM files are created
using a sparse flag therefore the space can be measured by counting number of
allocated blocks in the file. For Dash and CLHT, the files created are not
sparse. Therefore, we measure the space of the hash tables by computing the
minimum file size required by Dash and CLHT to complete the benchmark 
\rev{R4D16} {without complete doubling.} We start
with sizing the file equal to the size of the dataset and keep increasing the
size in increments of 100M until the benchmarks completes successfully.
We report the space usage as \emph{space efficiency} which is the ratio of the
size of the dataset over the size of the hash table.
All the instantaneous performance benchmarks are performed using a single
thread.

 
\para{System specification}
All the experiments were run on an Intel(R) Xeon(R) Gold 5218 CPU @ 2.30GHz with
two NUMA nodes, 16 cores per nodes, and 44M L3 cache. The machine has 192GiB of
DRAM running Linux kernel 5.4.0-70-generic. We restrict our runs to all the cores
on a single NUMA node to avoid NUMA effects in the performance.
For all the benchmarks, we increase the number of threads by powers of two
starting from 1 up to 16 (i.e., 1, 2, 4, 8, and 16) which is the maximum number
of cores on a NUMA node.

\para{PMEM setup}
The machine has 1536GiB of Intel Optane 100 series persistent memory in 12
128GiB DIMMs, 6 per socket. The PMEM is configured to use AppDirect mode and is
accessed using fsdax on an ext4 filesystem. This filesystem is configured with
a 2MiB stride to enable 2MiB huge page faults, and mounted using dax.
\sysname stores its data to PMEM by creating large sparse files at
initialization for each level, and only using (and therefore populating) a
prefix of each file.

\subsection{PMEM benchmarks}

\para{Micro benchmarks}
\Cref{fig:performance-pmem} shows the performance and scaling of \sysname, Dash, and
CLHT on microbenchmarks in PMEM.

\sysname always performs faster than Dash and CLHT. Specifically, it
is $1.1\times$--$2.7\times$ faster for insert, query, and remove
operations.

For all four operation types, all the hash tables scale almost
linearly. The scaling ratio (i.e., the ratio of the relative
throughput and the relative number of threads for a system) of
\sysname is 0.67, Dash is 0.56, and CLHT is 0.77.

\begin{table}[t]
   \small
   \resizebox{\linewidth}{!}{
   \begin{tabular}{| S |  r r r | r r r |}
      \toprule
      &  \multicolumn{3}{c|}{Insertions} & \multicolumn{3}{c|}{Positive Queries} \\
      \midrule
      \textbf{Percentile} & \textbf{\sysname} & \textbf{Dash} & \textbf{CLHT} & \textbf{\sysname} & \textbf{Dash} & \textbf{CLHT} \\
      \midrule
      50    & 353\,\si{\nano\second}   & 830\,\si{\nano\second}   & 1.29\,\si{\micro\second}  & 602\,\si{\nano\second}   & 834\,\si{\nano\second}   & 974\,\si{\nano\second}     \\
      95    & 1.11\,\si{\micro\second} & 2.39\,\si{\micro\second} & 2.63\,\si{\micro\second}  & 1.49\,\si{\micro\second} & 2.16\,\si{\micro\second} & 2.14\,\si{\micro\second}   \\
      99    & 1.97\,\si{\micro\second} & 3.50\,\si{\micro\second} & 3.72\,\si{\micro\second}  & 1.96\,\si{\micro\second} & 2.74\,\si{\micro\second} & 3.41\,\si{\micro\second}    \\
      99.9  & 249.88\,\si{\micro\second} & 78.4\,\si{\micro\second} & 5.68\,\si{\micro\second}  & 2.42\,\si{\micro\second} & 4.35\,\si{\micro\second} & 5.24\,\si{\micro\second}    \\
      99.99 & 277.52\,\si{\micro\second} & 103\,\si{\micro\second}  & 16.49\,\si{\micro\second} & 5.24\,\si{\micro\second} & 7.91\,\si{\micro\second} & 15.60\,\si{\micro\second}    \\
      max   & 37.09\,\si{\milli\second} & 8.62\,\si{\milli\second} & 12.31\,\si{\second}       & 259.65\,\si{\micro\second} & 16.0\,\si{\milli\second} &  153.21\,\si{\micro\second}  \\
      \bottomrule
   \end{tabular}
   }
   \caption{Percentile latencies in \sysname, Dash and CLHT for YCSB workload A run on PMEM using 16 threads.}
   \label{tab:latency-dist-pmem}
   \vspace{-3em}
\end{table}

\para{YCSB workloads}
\Cref{fig:performance-pmem} shows the performance
of \sysname, Dash, and CLHT for three YCSB workloads on PMEM.

For the load phase of these workloads, \sysname is faster than other hash
tables. Specifically, it is between $1.1\times$ and $2.5\times$ faster than Dash
and CLHT. For the run phase all three workloads, \sysname is faster compared to
both Dash and CLHT. CLHT performance for workload C is closer to \sysname.
Workload C consists of 100\% queries. And this observation is consistent with
the positive query performance in microbenchmarks.

The YCSB benchmarks show that \sysname performs better than other hash tables
when the workload also involves resizing the hash table as the YCSB load phase
and workloads A and B require the hash tables to resize at least twice. 
Moreover, similar to the microbenchmarks, the load performance of \sysname
scales almost linearly with increasing number of threads.

For different workload types (A, B, and C), the performance of \sysname is
always better than other hash tables and also scales almost linearly with
increasing number of threads.

\rev{R4D2, R4D3} { \para{Discussion} The high performance of \sysname both
  on the micro and YCSB workloads is primarily due to the small number of PMEM
  accesses during insert, query, and delete operations.  During insert
  and delete operations, we only perform a single PMEM write. During
  query operations, we usually perform at most a single PMEM read
  (unless there is a false positive in the metadata).  Furthermore,
  since most items are in level 1, most inserts, deletes, and positive
  queries access only a single DRAM cache line, as well.  Negative
  queries must access 4 DRAM cache lines (1 metadata cache line for
  level 1, 2 for level 2, and 1 for level 3), but they usually do not
  have to access a PMEM cache line at all.  Finally, metadata searches
  are implemented using vector instructions, so they take constant time
  even though our buckets are larger than a cache line.
%
}

\para{Insert and query latency}
%
\Cref{tab:latency-dist-pmem} shows the 50, 95, 99,
99.9, and 99.99 percentiles and the worst case for insert and positive query
operations in the benchmarked hash tables.

On PMEM, Dash has slower latency up to 99.99 percentile compared to \sysname for
both inserts and queries. However, Dash is $2\times$ faster for the worst-case
insert latency and about 50\% slower for the worst-case query latency.

CLHT has the worst-case insert latency of 12 seconds. This is because during a
resize operation all active inserts are stopped and insert threads help to move
the keys from the old hash table to the new one.
In CLHT, the query latency is always good. This is because the queries can
always perform probes on the old copy of the hash table even when the resize is
active. Queries are never blocked in CLHT.
CLHT performs resizes by allocating a new hash table of twice the size and moving
key-value pairs from the old hash table to the new one. 

The latency of operations is computed during the YCSB workload A run that
contains insert and positives queries (50/50). The workload is configured so
that hash tables must perform at least one resize during the run. All the hash
tables are run using 16 threads.
Comparing the latency of operations during a workload run helps explain the
impact of a resize on the worst case latency of operations.

\begin{table}[t]
\begin{tabular}{| l | r |}
\toprule
  \textbf{Hash table} & \textbf{Space Efficiency} \\
\midrule
\sysname  & 85\% \\
Dash      & 69\% \\
CLHT      & 33\% \\
\bottomrule
\end{tabular}
\caption{Space efficiency of PMEM hash tables. Space efficiency is
  the ratio of Data size over hash table size. We compute the space efficiency
   after inserting $0.95N$ keys-value pairs in the hash table where $N$ is the
  initial capacity.}
  \label{tab:intro-space}
  \vspace{-3em}
\end{table}

\rev{R4D1}{
\para{Space efficiency in PMEM}
\Cref{tab:intro-space} shows the
space efficiency of PMEM hash tables. Both Dash and CLHT have low space
efficiency compared to \sysname. \sysname PMEM representation is 1.2GB for a
dataset size of 1.06GB ($2^{26}*1.07$ 8 Byte keys and values) and in-memory
representation is $~\approx80$MB.
}

\begin{figure*}[t]
  \centering
  \ref{dram-bar-legend}
  \begin{subfigure}{0.24\linewidth}
  \begin{tikzpicture}
    \begin{axis}[
        IntroBarPlotWrite,
        legend entries = {IcebergHT, Cuckoo, TBB, CLHT},
        legend columns = 4,
        legend cell align = {left},
        legend to name={dram-bar-legend}
      ]
      \addplot[IcebergStyle, fill]    table {data/iceberg_micro_dram_bar_write_blocklock.csv};
      \addplot[CuckooStyle, fill]    table {data/cuckoo_micro_dram_bar_write.csv};
      \addplot[TBBStyle, fill]    table {data/tbb_micro_dram_bar_write.csv};
      \addplot[CLHTStyle, fill]    table {data/clht_lb_res_micro_dram_bar_write.csv};
    \end{axis}
  \end{tikzpicture}
\end{subfigure}
  \hspace{10pt}
\begin{subfigure}{0.24\linewidth}
  \begin{tikzpicture}
    \begin{axis}[
        IntroBarPlotRead,
      ]
      \addplot[IcebergStyle, fill]    table {data/iceberg_micro_dram_bar_read_blocklock.csv};
      \addplot[CuckooStyle, fill]    table {data/cuckoo_micro_dram_bar_read.csv};
      \addplot[TBBStyle, fill]    table {data/tbb_micro_dram_bar_read.csv};
      \addplot[CLHTStyle, fill]    table {data/clht_lb_res_micro_dram_bar_read.csv};
    \end{axis}
  \end{tikzpicture}
\end{subfigure}
%
  \begin{subfigure}{0.24\linewidth}
  \begin{tikzpicture}
    \begin{axis}[
        YcsbBarPlotLa,
      ]
      \addplot[IcebergStyle, fill]    table {data/iceberg_ycsb_dram_bar_la_blocklock.csv};
      \addplot[CuckooStyle, fill]    table {data/cuckoo_ycsb_dram_bar_la.csv};
      \addplot[TBBStyle, fill]    table {data/tbb_ycsb_dram_bar_la.csv};
      \addplot[CLHTStyle, fill]    table {data/clht_lb_res_ycsb_dram_bar_la.csv};
    \end{axis}
  \end{tikzpicture}
\end{subfigure}
  \hspace{-10pt}
\begin{subfigure}{0.24\linewidth}
  \begin{tikzpicture}
    \begin{axis}[
        YcsbBarPlotBc,
      ]
      \addplot[IcebergStyle, fill]    table {data/iceberg_ycsb_dram_bar_bc_blocklock.csv};
      \addplot[CuckooStyle, fill]    table {data/cuckoo_ycsb_dram_bar_bc.csv};
      \addplot[TBBStyle, fill]    table {data/tbb_ycsb_dram_bar_bc.csv};
      \addplot[CLHTStyle, fill]    table {data/clht_lb_res_ycsb_dram_bar_bc.csv};
    \end{axis}
  \end{tikzpicture}
\end{subfigure}
  \caption{Throughput for insertions, deletions, and queries (positive and
  negative) using 16 threads for DRAM hash tables. The throughput is
  computed by inserting $0.95N$ keys-value pairs where $N$ is the
  initial capacity of the hash table. (Throughput is Million ops/second)}
  \label{fig:throughput-dram-bar}
\end{figure*}
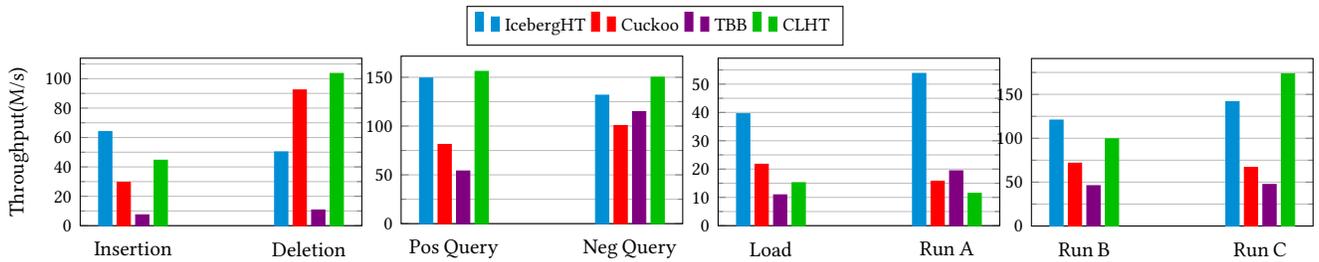

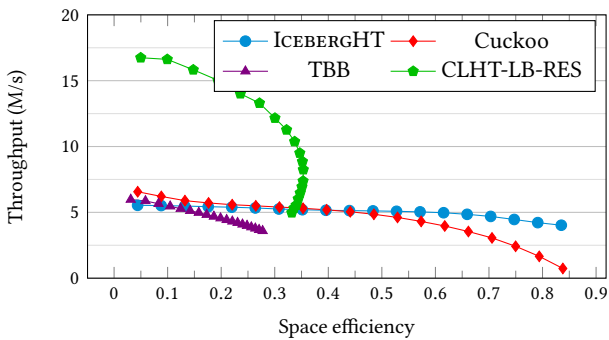
\begin{figure}
   \centering
      \begin{tikzpicture}
         \begin{axis}[
               InstanSePlot,
               legend entries={\color{black}\sysname, \color{black}\cuckoo, \color{black}\tbb, \color{black}\clht},
               legend columns=2
            ]
            \addplot[IcebergStyle] table[x=se, y=throughput] {data/iceberg_micro_dram_instan_se_blocklock.csv};
            \addplot[CuckooStyle]  table[x=se, y=throughput] {data/cuckoo_micro_dram_instan_se.csv};
            \addplot[TBBStyle]     table[x=se, y=throughput] {data/tbb_micro_dram_instan_se.csv};
            \addplot[CLHTStyle]    table[x=se, y=throughput] {data/clht_lb_res_micro_dram_instan_se.csv};
         \end{axis}
      \end{tikzpicture}
  \caption{Insertion throughput and space efficiency performance of hash tables
  in DRAM. (Throughput is Million ops/second)}
\label{fig:performance-micro-dram-instan-se}
  \vspace{-3pt}
\end{figure}

\subsection{DRAM performance}

\para{Micro and YCSB benchmarks.}
\Cref{fig:throughput-dram-bar} shows the performance of \sysname, cuckoo, TBB,
and CLHT on microbenchmarks and YCSB workloads using 16 threads in
DRAM.

\sysname is $2.3\times$--$9.1\times$ faster for insertions and
$1.7\times$--$2.6\times$ faster for lookups than the libcuckoo and TBB. For deletions, \sysname is up to $5.3\times$ faster than TBB but
$\approx50\%$ slower than libcuckoo.
\sysname is also faster than CLHT for insertions. However, CLHT has faster
deletions and query operations compared to \sysname. This is due the extra
overhead of one metadata probe in level 1 and two probes in level 2 in \sysname
in DRAM. These metadata probes are essential to avoid multiple cache line access
in the main table, especially on PMEM where accessing multiple locations in the
table can hurt performance.

\Cref{fig:throughput-dram-bar} shows the performance of \sysname and other hash
tables for YCSB workloads.
For the load phase of these workloads, \sysname is faster than other hash
tables. It is up to $2.2\times$ faster than libcuckoo, $4.4\times$ faster than
TBB, and $2.9\times$ faster then CLHT in DRAM. For workload C which contains all
queries, CLHT is faster than \sysname. This is similar to the query
workload results in the microbenchmarks.

The faster query performance of CLHT comes at a high space overhead.
Specifically, CLHT uses $3\times$ more space than \sysname.

\para{Insert and query latency in DRAM}
\Cref{tab:latency-dist-dram} shows the 50, 95, 99, 99.9, and 99.99 percentiles
and the worst case for insert and positive query operations in various hash
tables in DRAM.
The latency of operations is computed in the same way as it was done for the
PMEM benchmarks.

\sysname and libcuckoo have similar median insert latency but the worst case
latency is three orders of magnitude slower in libcuckoo. This is due to the
fact that \sysname performs resizing in a lazy dynamic manner which helps
to avoid stalling other operations during a big resize.
TBB's median insert latency is $2\times$ higher than \sysname and libcuckoo. But
TBB's worst-case latency is an order of magnitude faster than libcuckoo. This is
because resizes can be done fairly efficiently by splitting buckets in TBB and do
not require a complete rehashing of items.

libcuckoo has the lowest median query latency compared to \sysname and TBB. However,
the worst-case latency is again about three orders of magnitude slower than
\sysname.
TBB has the lowest worst-case query latency due to the fact the splitting a
bucket is fairly fast and can be achieved using a pointer swing. However, in
\sysname a few queries may have to wait if the block they want to look into is
getting fixed during a resize.

\subsection{Speed/space tradeoff}
\Cref{fig:performance-micro-dram-instan-se} shows the instantaneous
DRAM insertion throughput of \sysname, libcuckoo, TBB, and CLHT versus
their space efficiency . We compare instantaneous throughput versus
space efficiency only in DRAM only because it is not always possible
to measure the instantaneous space usage of PMEM-based hash tables
(see discussion above), whereas in DRAM we can always get the
MaxRSS.  The point of these experiments is to uncover the general
relationship between insertion performance and space usage.

As \Cref{fig:performance-micro-dram-instan-se} shows, CLHT's insertion
performance in DRAM comes at a high price in terms of space
efficiency.  CLHT never gets a space efficiency higher than 40\%.

\rev{R5O6}{CLHT space efficiency improves initially as the 3-entry bucket-heads
fill but then begins to decline as bucket-heads overflow, necessitating the
allocation of 3-entry overflow links in its chains.
In~\Cref{fig:performance-micro-dram-instan-se}, the change in the
space efficiency of the CLHT is marginal after 30\% and therefore these points
are clustered together.}

\Cref{fig:performance-micro-dram-instan-se} also shows that
\sysname offers both high space efficiency and high insertion throughput.
\sysname also has consistent insertion throughput irrespective of the
space usage. Interestingly, the throughput increases (beyond
80\%) as more keys end up in level 2 and 3. For example, going from 85\% to 90\%
load, $\approx47\%$ of the keys end up in level 2 and from 90\% to 95\% load,
almost $65\%$ keys end up in level 2.
Inserting keys in level 2 is comparatively faster than level 1 as level 2 is
much smaller in size compared to level 1. Due to the smaller size, a major
fraction of the level 2 can be cached in the last level cache (LLC).

The insertion throughput for both libcuckoo and TBB drops as the
space efficiency increases. For libcuckoo, the drop in the throughput is fairly sharp
above 70\% space efficiency. For TBB, the drop is consistent and gradual up to 95\%
space efficiency.

\begin{table}[t]
\begin{tabular}{| l | r r r |}
\toprule
  \textbf{Benchmark} & \textbf{Level 1} & \textbf{Level 2} & \textbf{Level 3} \\
\midrule
Micro           & 91.2\% & 8.7\% & 0.000082\% \\
YCSB load       & 95.9\% & 4.0\% & 0\% \\
YCSB Workload A & 95.8\% & 4.1\% & 0\% \\
YCSB Workload B & 95.8\% & 4.1\% & 0\% \\
\bottomrule
\end{tabular}
\caption{Distribution of keys across the three levels in \sysname hash table.}
  \label{tab:key-dist}
   \vspace{-3em}
\end{table}

\subsection{Distribution of keys in \sysname}
\Cref{tab:key-dist} shows the distribution of keys across the three levels in
\sysname. Most of the keys (>90\%) reside in level 1 across all the  benchmarks
and workloads. A small percentage of keys (<10\%) reside in
level 2 and almost no keys are found in level 3. This shows that the empirical
distribution of keys across different levels follows the theoretical guarantees
of Iceberg hashing.

Level 3 sees a tiny number of keys in the microbenchmark because, in the
microbenchmarks, we fill the table to 95\% load factor without resizing.
However, even at 95\% load factor, the number of keys in level 3 is negligible
and does not impact the query or deletion performance.

For YCSB workloads, we report the distribution after the load phase (which is
the same across the three workloads) and also after the run phase for workloads A
and B that contain new insertions. The \sysname hash table has default load
factor threshold of 85\% which means a resize is invoked when the hash table
reaches an 85\% load factor. This makes the hash table always have enough
space in levels 1 and 2 so level 3 remains empty.

\begin{table}
   \small
   \resizebox{\linewidth}{!}{
   \begin{tabular}{| S | r r r | r r r |}
      \toprule
      & \multicolumn{3}{c|}{Insertions} & \multicolumn{3}{c|}{Positive Queries} \\
      \midrule
      \textbf{Percentile} & \textbf{\sysname} & \textbf{libcuckoo} & \textbf{TBB} & \textbf{\sysname} & \textbf{libcuckoo} & \textbf{TBB} \\
      \midrule
      50    & 336\,\si{\nano\second}   & 264\,\si{\nano\second}   & 819\,\si{\nano\second}   & 290\,\si{\nano\second}   & 198\,\si{\nano\second}   & 494\,\si{\nano\second}   \\
      95    & 671\,\si{\nano\second}   & 2.02\,\si{\micro\second} & 1.59\,\si{\micro\second} & 548\,\si{\nano\second}   & 429\,\si{\nano\second}   & 955\,\si{\nano\second}   \\
      99    & 1.09\,\si{\micro\second} & 5.99\,\si{\micro\second} & 2.24\,\si{\micro\second} & 687\,\si{\nano\second}   & 562\,\si{\nano\second}   & 1.22\,\si{\micro\second}  \\
      99.9  & 22.03\,\si{\micro\second} & 19.8\,\si{\micro\second} & 6.52\,\si{\micro\second} & 979\,\si{\nano\second} & 836\,\si{\nano\second}   & 1.57\,\si{\micro\second}  \\
      99.99 & 29.08\,\si{\micro\second} & 219\,\si{\micro\second}  & 9.27\,\si{\micro\second} & 1.93\,\si{\micro\second} & 218\,\si{\micro\second}  & 4.97,\si{\micro\second}    \\
      max   & 345.34\,\si{\milli\second} & 2.05\,\si{\second}       & 734\,\si{\milli\second}  & 38.35\,\si{\micro\second} & 1.01\,\si{\second}       & 42.8\,\si{\micro\second}   \\
      \bottomrule
   \end{tabular}
   }
   \caption{Percentile latencies in \sysname, libcuckoo and TBB for YCSB workload A run on DRAM using 16 threads.}
   \label{tab:latency-dist-dram}
   \vspace{-3em}
\end{table}

\rev{R4D1}{
\subsection{Configuring front yard and back yard}
\Cref{tab:config-frontback-dram} shows the performance of \sysname with
different block sizes in front and back yards. The goal of these experiments is
to determine the best configuration of front and back yard to achieve high
performance and fill capacity. We vary the block sizes in front and back yard
and fill up each instance to 95\% load factor and evaluate the performance.

Reducing the number of blocks in L2 to 6 results in more items going into L3.
This results in faster operations overall. However, reducing the L2 blocks to 4
slows down the negative queries considerably due to a high fraction of items in
L3 which require pointer chasing during queries.
Reducing the block size in L1 to 32 increases the fraction of items going into
L2 and L3. This results in slowdown across the board.
This also means that if we size front and back yards equally then the
performance would be worse as more items would end up in L2/L3 causing extra
cache  misses.
}

\begin{table}
   \small
   \resizebox{\linewidth}{!}{
   \begin{tabular}{| c | r r r r | r r|}
      \toprule
      Block size & Insertions & Neg Queries & Pos queries & Deletions & \%L2 & \%L3 \\
      \midrule
      L1 64 L2 8  & 62.94 & 128.71 & 144.23 & 50.48  & 8.7   & 0.00007 \\
      L1 64 L2 6  & 65.58 & 129.27 & 149.27 & 53.77  & 7.0   & 0.007   \\
      L1 64 L2 4  & 64.36 & 115.07 & 152.12 & 51.60  & 5.4   & 0.27    \\
      L1 32 L2 8  & 53.99 & 109.28 & 129.54 & 45.97  & 17.7  & 0.03    \\
      L1 32 L2 6  & 54.75 & 109.15 & 133.71 & 49.31  & 13.7  & 0.06    \\
      L1 32 L2 4  & 53.20 & 95.99  & 140.08 & 46.33  & 10.38 & 0.64    \\
      \bottomrule
   \end{tabular}
   }
   \caption{Performance of \sysname for different front/backyard block sizes on
   DRAM using 16 threads. Throughput in Million/sec. Each instance is filled to
   95\% capacity.}
   \label{tab:config-frontback-dram}
   \vspace{-3em}
\end{table}

\section{Related work}\label{sec:related}

In this section, we will discuss various hash table implementations and their
applications. A discussion of various hash table designs used in our evaluation
is given in~\Cref{sec:other-ht}.  



\para{In-memory hash tables}
There are numerous in-memory hash table implementations such as sparse and dense
hash maps from Google~\cite{googlesparse}, the F14 hash table from
Facebook~\cite{F14}, the FASTER hash table from Microsoft~\cite{ChandramouliPrKo18},
the hash table in Intels' TBB library~\cite{Pheatt08}, the cuckoo hash
table~\cite{LiAn14}, the linear probing-based fast hash
table~\cite{maier2019dynamic,MaierSaDe19}, and the unordered map in C++ STL.
However, most of these hash tables only support single threaded operations.

MemC3~\cite{fan2013memc3} supports multiple readers but only a single writer. It
is based on optimistic concurrent cuckoo hashing. MemC3 also supports
variable-length keys and optimizes accesses using fingerprinting.
FASTER~\cite{ChandramouliPrKo18} further optimizes the implementation by storing
the tag in the higher order bits of the pointer. It also supports scaling out of
memory to a secondary storage device and supports crash safety using logging.
Libcuckoo~\cite{LiAn14} extends MemC3 to support multiple readers and writers.

\para{Persistent-memory hash tables}
Persistent memory offers byte-addressability and high capacity compared
to other traditional storage mediums. This makes PMEM an attractive medium for
building dynamic hash tables. 
Recently numerous hash tables have been developed for
PMEM~\cite{LuHaWa20,level-hashing,cceh,nvc-hashmap,LeeEtAl19-Recipe,phprx,pfht}.
The main goal of PMEM-based hash tables is to reduce the number of write
operations during an insert/remove while still support efficient queries.

PFHT~\cite{pfht} reduces the number of writes using a two-level scheme similar
to \sysname where the second level acts as a stash (or backyard). Similar to
level 3 in \sysname PFHT also uses linked lists to store items in the stash.
Path hashing~\cite{zuo2017write} optimizes the storage in the stash by
reorganizing it into a tree structure. This lowers the search costs in the
stash.
Level hashing~\cite{level-hashing,zuo2018write} is another two-level scheme that
bounds the search cost to at most four buckets.

CCEH~\cite{cceh} is based on extendible hashing~\cite{fagin1979extendible}. It
is crash-consistent and the extendible design helps to avoid rehashing all the
items after a resize. The queries tend to be slower due to random memory access.
Therefore, it bounds the probing length to a few cachelines but that in turn
leads to low load factors.
NVC-hashmap~\cite{nvc-hashmap} presents a lock-free design for a PMEM-based hash
table. The lock-free design though suitable for PMEM has added implementation
complexity and makes searching slower due to pointer chasing.

\para{Applications}
Hash tables are widely used to maintain symbol tables in compilers, implement
caches, index databases, manage memory pages in Linux, implement routing
tables, and to build inverted indexes for document search.
Examples of such systems are Redis~\cite{redis}, Memcached~\cite{memcached},
Cassandra~\cite{cassandra}, DynamoDB~\cite{dynamo}, MongoDB~\cite{mongo}, etc.
These implementations have been further improved in follow up works such as
MemC3~\cite{fan2013memc3}, MICA~\cite{lim2014mica}, and SILT~\cite{lim2011silt}.


\section{Discussion}\label{sec:discussion}
%
We attribute the high performance and space-efficiency to stability and low
associativity. Stability helps in achieving a faster inserts. Low associativity
helps in getting faster query performance. Iceberg hashing achieves both
stability and low associativity at the same time.


\sysname insertion performance with 16 threads is about 70\% of the hardware
limit. The 30\% overhead in the insert operation is due the overhead of
maintaining transient information, e.g., to update the metadata and increment
counters for resize checks. We were able to get to 85\% of the hardware limit by
commenting out counter-maintenance code and using huge pages.
For query performance, the overhead is about 50\%. Some of this overhead is due
to the same factors as in the insert operation. However, the query operation has
other overheads that results in extra PMEM access. For example, there is a
25\% chance of a collision in the metadata fingerprints that results in
extra PMEM accesses during the query operation.

In the DRAM setting (where the hash table resides in DRAM), the cost of metadata
accesses is a non-trivial fraction of the overall operation cost. Therefore, each query
operation incurs at least two cache line misses. CLHT on the other hand performs
a single cache line miss for most of the keys.  This results in \sysname having
a slightly slower query and deletion performance compared to CLHT.

\rev{R4D15, R5O4}{ Our implementation supports 8-byte keys and 8-byte
  values.  As in other hash-table designs, such as Dash, this core
  functionality can be extended to variable-length keys and values by
  storing pointers to the actual keys and values in the hash table.
}




\section*{Acknowledgments}

We gratefully acknowledge support from NSF grants
CCF 805476, 
CCF 822388,   
CCF 1724745,  
CCF 1715777,  
CCF 1637458,  
IIS 1541613,  
CNS 2118620,  
CNS 1938180,  
CCF 2106999,  
CNS 1408695, 
CNS 1755615, 
CCF 1439084, 
CCF 1725543, 
CSR 1763680, 
CCF 1716252, 
CCF 1617618, 
CNS 1938709, 
IIS 1247726, 
CNS-1938709.   
%
%
%
%
%
%
Kuszmaul was funded by a John and Fannie Hertz Fellowship. Kuszmaul was also
partially sponsored by the United States Air Force Research Laboratory and the
United States Air Force Artificial Intelligence Accelerator and was accomplished
under Cooperative Agreement Number FA8750-19-2-1000. The views and conclusions
contained in this document are those of the authors and should not be
interpreted as representing the official policies, either expressed or implied,
of the United States Air Force or the U.S. Government. The U.S. Government is
authorized to reproduce and distribute reprints for Government purposes
notwithstanding any copyright notation herein.


\bibliographystyle{ACM-Reference-Format}

\bibliography{bibliography}
\end{document}